\documentclass[runningheads]{./llncs/llncs}

\usepackage{mathrsfs}
\usepackage{amsmath}
\usepackage{amssymb}
\usepackage{xspace}

\usepackage{petriGames}
\usepackage{ltl}
\usepackage{hyper-LTL}

\usepackage{wasysym}

\usepackage{wrapfig}
\usepackage{float}

\usepackage{tikz}
\usetikzlibrary{petri,positioning}

\usepackage{color}
\usepackage{xcolor}
\usepackage[draft]{todonotes}

\usepackage{ulem}
\usepackage{bytefield}

\usepackage{longtable}

\usepackage{graphicx}

\makeatletter
\RequirePackage[bookmarks,unicode,colorlinks=true]{hyperref}%
   \def\@citecolor{blue}%
   \def\@urlcolor{blue}%
   \def\@linkcolor{blue}%

\def\orcidID#1{\smash{\href{http://orcid.org/#1}{\protect\raisebox{-1.25pt}{\protect\includegraphics{./llncs/ORCID_Color.eps}}}}}
\makeatother

\def\sindex#1{}

\def\mtext#1{\mbox{~{#1}~}}

\def\calN#1{\mbox{$\cal N_{\mbox{\scriptsize{#1}}}\hskip .3mm$\/}}
\def\scalN{\mbox{\scriptsize \calN{}}}

\def\calF#1{\mbox{$\cal F_{\mbox{\scriptsize{#1}}}$\/}}
\def\calFN{\mbox{${\cal F}_{\scalN}$\/}}

\def\calP#1{\mbox{$\cal P_{\mbox{\scriptsize{#1}}}$\/}}

\def\calT#1{\mbox{$\cal T_{\mbox{\scriptsize{#1}}}\hskip .3mm$\/}}

\def\rmA#1{\mbox{${A}_{\mbox{\scriptsize{#1}}}$\/}}

\def\rmI#1{\mbox{${I}_{\mbox{\scriptsize{#1}}}$\/}}

\def\rmM#1{\mbox{${M}_{\mbox{\scriptsize{#1}}}$\/}}

\def\rmO#1{\mbox{${O}_{\mbox{\scriptsize{#1}}}$\/}}

\def\rmPl#1{\mbox{${Pl}_{\mbox{\scriptsize{#1}}}$\/}}

\def\rma#1{\mbox{${a}_{\mbox{\scriptsize{#1}}}$\/}}

\def\rmp#1{\mbox{${p}_{\mbox{\scriptsize{#1}}}$\/}}
\def\rmq#1{\mbox{${q}_{\mbox{\scriptsize{#1}}}$\/}}
\def\rmr#1{\mbox{${r}_{\mbox{\scriptsize{#1}}}$\/}}

\def\rmt#1{\mbox{${t}_{\mbox{\scriptsize{#1}}}$\/}}
\def\rmu#1{\mbox{${u}_{\mbox{\scriptsize{#1}}}$\/}}

\def\mathf#1{\mbox{$f_{\mbox{\scriptsize{#1}}}$\/}}

\def\arrow{\mbox{|$\!\rightarrow$\/}}
\def\larrow{\mbox{|$\!$|$\!\rightarrow$\/}}

\def\arrowu#1{\mbox{${\stackrel{\mbox{#1}}{\larrow}}$\/}}

\def\arrowd#1{\mbox{${\arrow}_{\mbox{\scriptsize{#1}}}$\/}}

\def\init{\mbox{\it init}}
\def\act{\mbox{\it act}}

\def\pre{\mbox{\it pre}}

\def\post{\mbox{\it post}}
\def\reach{\mbox{\it reach}}

\begin{document}
\title{Concurrent Hyperproperties}

\titlerunning{Concurrent Hyperproperties}

\author{
Bernd Finkbeiner\inst{1} \and
Ernst-R\"udiger Olderog\inst{2}
}

\authorrunning{
B.\ Finkbeiner and E.-R. Olderog
} 
%
\institute{
CISPA Helmholtz Center for Information Security, Saarbr\"ucken, Germany\\
ORCID ID: 0000-0002-4280-8441 \\
\email{finkbeiner@cispa.de} 
\and
Carl von Ossietzky University of Oldenburg, Oldenburg, Germany\\
 ORCID ID: 0000-0002-3600-2046 \\
\email{olderog@informatik.uni-oldenburg.de}
}

\maketitle

\begin{abstract}
Trace properties, which are sets of execution traces,
are often used to analyze systems, but their expressiveness is limited.
Clarkson and Schneider defined \emph{hyperproperties} 
as a generalization of trace properties to sets of sets of traces. 
Typical applications of hyperproperties are found in information flow security.
We introduce an analogous definition of \emph{concurrent} hyperproperties, 
by generalizing traces to \emph{concurrent} traces, 
which we define as partially ordered multisets. 
We take Petri nets as the basic semantic model.
Concurrent traces are formalized via causal nets.
To check concurrent hyperproperties, we define \emph{may} and \emph{must testing} 
of sets of concurrent traces in the style of DeNicola and Hennessy, 
using the parallel composition of Petri nets.
In our approach, we thus distinguish nondeterministic and concurrent behavior.
We discuss examples where concurrent hyperproperties are needed.
\end{abstract}


\section{Introduction}
\label{sec:introduction}

Among the most fundamental debates in the theory of concurrency is the
distinction between \emph{interleaving} semantics in the style of
Milner~\cite{milner_calculus_1980} and Hoare~\cite{Hoare80a}, and
\emph{partial-order} (or \emph{true concurrency}) semantics following
the work of Petri~\cite{Pet77}, Mazurkiewicz~\cite{Mazurkiewicz_1977},
and Winskel~\cite{UCAM-CL-TR-95}. In interleaving semantics,
concurrency is reduced to its sequential nondeterministic simulation;
in partial-order semantics, concurrency is modeled as causal
independence.

In this paper, we revisit this classic debate in the modern setting of
\emph{hyperproperties}. Clarkson and Schneider defined hyperproperties
as a generalization of trace properties, which are sets of traces, to
\emph{sets of} sets of traces~\cite{ClaSch10}.  
Hyperproperties are a powerful class of linear-time properties that can express many notions related to information flow, symmetry, robustness, and causality. A typical example is \emph{noninterference}~\cite{GoguenMeseguer:1982:NI}, which is one of the most well-studied information-flow security policies. Noninterference requires that for all computations and for all sequences of actions of a high-security agent $A$, the resulting observations made by a low-security observer $B$ are identical to $B$’s observations that would result without $A$’s actions.
While trace properties express properties of individual
executions, hyperproperties express properties of sets of traces.
This makes it possible to relate different executions, for example by
requiring that certain observations are the same, without necessarily
restricting the events on individual executions.

Since hyperproperties refer to traces, they are, at least in
principle, immediately applicable to concurrent systems with
interleaving semantics. However, the interleaving semantics leads to a
fundamental problem, which we will illustrate with a sequence of
example systems given as the Petri nets shown in Fig.~\ref{fig:example-systems}.
We employ the usual graphical representation of Petri nets:
circles represent places and boxes represent transitions that are connected to places via
directed arcs. In our setting, transitions are labeled by action symbols like $h_1$ and $h_2$.
Black dots represent tokens, which represent the current points of
activity. The simultaneous presence of several tokens models concurrent activities.
The dynamic behavior of a Petri net is modeled by its token game
that defines how tokens can move inside the net. A transition is enabled
if all places connected to it with an ingoing arc carry a token.
Firing the transition moves these tokens to the places connected
to it with an outgoing arc. Branching from a place models nondeterministic choice,
whereas branching from a transition models the start of a concurrent execution.
As an example, consider the net $\mathcal N_C$ shown on the right in Fig.~\ref{fig:example-systems}.
From the initial place $p_0$, there is a nondeterministic choice between the transitions
labeled with $h_1$ and $h_2$. Firing transition $h_1$ concurrently enables the
transitions labeled with $l_1$ and $l_2$, whereas firing transition $h_2$ enables in place $p_{13}$
the nondeterministic choice between the transitions $l_1$ and $l_2$.
For more details on Petri nets we refer to Section~\ref{sec:nets}.

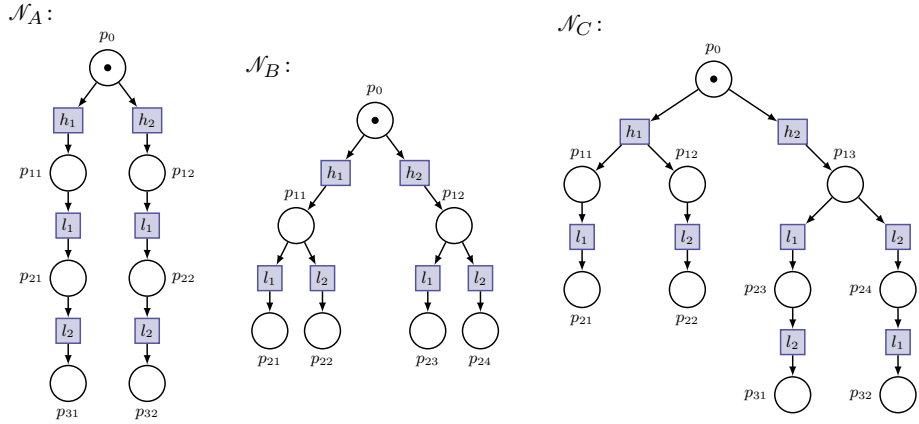
\begin{figure}[t]

\begin{minipage}{2.5cm}

$\calN{$A$}$:

\scalebox{0.7}{
  \begin{tikzpicture}

\node[place,
    tokens=1,
    label=$p_0$] (p0) at (2.5,6) {};
 
\node[place,
    label=left:$p_{11}$] (p11) at (1.75,4) {};
    
\node[place,
    label=right:$p_{12}$] (p12) at (3.25,4) {};
    
\node[place,
    label=left:$p_{21}$] (p21) at (1.75,2) {};
    
\node[place,
    label=right:$p_{22}$] (p22) at (3.25,2) {};

\node[place,
    label=below:$p_{31}$] (p31) at (1.75,0) {};
    
\node[place,
    label=below:$p_{32}$] (p32) at (3.25,0) {};

\node[transition,
    label=] (t01) at (1.75,5) {$h_1$};
       
\node[transition,
    label=] (t02) at (3.25,5) {$h_2$};
 
\node[transition,
    label=] (t11) at (1.75,3) {$l_1$};
 
\node[transition,
    label=] (t12) at (3.25,3) {$l_1$};
    
\node[transition,
    label=] (t21) at (1.75,1) {$l_2$};
 
\node[transition,
    label=] (t22) at (3.25,1) {$l_2$};
 
\draw[-latex,thick] (p0) -- (t01);
\draw[-latex,thick] (p0) -- (t02);

\draw[-latex,thick] (p11) -- (t11);
\draw[-latex,thick] (p12) -- (t12);
\draw[-latex,thick] (p21) -- (t21);
\draw[-latex,thick] (p22) -- (t22);
 
\draw[-latex,thick] (t01) -- (p11);
\draw[-latex,thick] (t02) -- (p12);
\draw[-latex,thick] (t11) -- (p21);
\draw[-latex,thick] (t12) -- (p22);
\draw[-latex,thick] (t21) -- (p31);
\draw[-latex,thick] (t22) -- (p32);
 
\end{tikzpicture}
}	
\end{minipage}
\begin{minipage}{3.5cm}

$\calN{$B$}$:

\scalebox{0.7}{
  \begin{tikzpicture}

\node[place,
    tokens=1,
    label=$p_0$] (p0) at (2.5,6) {};
 
\node[place,
    label=$p_{11}$] (p11) at (1,4) {};
    
\node[place,
    label=$p_{12}$] (p12) at (4,4) {};
    
\node[place,
    label=below:$p_{21}$] (p21) at (0.5,2) {};
\node[place,
    label=below:$p_{22}$] (p22) at (1.5,2) {};

\node[place,
    label=below:$p_{23}$] (p23) at (3.5,2) {};

\node[place,
    label=below:$p_{24}$] (p24) at (4.5,2) {};
    
\node[transition,
    label=] (t01) at (1.75,5) {$h_1$};
       
\node[transition,
    label=] (t02) at (3.25,5) {$h_2$};
 
\node[transition,
    label=] (t11) at (0.5,3) {$l_1$};
 
\node[transition,
    label=] (t12) at (1.5,3) {$l_2$};
    
\node[transition,
    label=] (t13) at (3.5,3) {$l_1$};
 
\node[transition,
    label=] (t14) at (4.5,3) {$l_2$};
 
\draw[-latex,thick] (p0) -- (t01);
\draw[-latex,thick] (p0) -- (t02);

\draw[-latex,thick] (p11) -- (t11);
\draw[-latex,thick] (p11) -- (t12);
\draw[-latex,thick] (p12) -- (t13);
\draw[-latex,thick] (p12) -- (t14);
 
\draw[-latex,thick] (t01) -- (p11);
\draw[-latex,thick] (t02) -- (p12);
\draw[-latex,thick] (t11) -- (p21);
\draw[-latex,thick] (t12) -- (p22);
\draw[-latex,thick] (t13) -- (p23);
\draw[-latex,thick] (t14) -- (p24);
 
\end{tikzpicture}
}	
\end{minipage}
\begin{minipage}{4.9cm}

$\calN{$C$}$:

\scalebox{0.7}{
  \begin{tikzpicture}

\node[place,
    tokens=1,
    label=$p_0$] (p0) at (2.5,6) {};
 
\node[place,
    label=$p_{11}$] (p11) at (0,4) {};

\node[place,
    label=$p_{12}$] (p12) at (2,4) {};
 
\node[place,
    label=$p_{13}$] (p13) at (5,4) {};

\node[place,
    label=below:$p_{21}$] (p21) at (0,2) {};
    
\node[place,
    label=below:$p_{22}$] (p22) at (2,2) {};
    
\node[place,
    label=left:$p_{23}$] (p23) at (4,2) {};

\node[place,
    label=left:$p_{24}$] (p24) at (6,2) {};

\node[place,
    label=left:$p_{31}$] (p31) at (4,0) {};
    
\node[place,
    label=left:$p_{32}$] (p32) at (6,0) {};

\node[transition,
    label=] (t01) at (1,5) {$h_1$};

\node[transition,
    label=] (t02) at (4,5) {$h_2$};

\node[transition,
    label=] (t11) at (0,3) {$l_1$};

\node[transition,
    label=] (t12) at (2,3) {$l_2$};
    
\node[transition,
    label=] (t13) at (4,3) {$l_1$};
 
\node[transition,
    label=] (t14) at (6,3) {$l_2$};
 
\node[transition,
    label=] (t21) at (4,1) {$l_2$};
 
\node[transition,
    label=] (t22) at (6,1) {$l_1$};

\draw[-latex,thick] (p0) -- (t01);
\draw[-latex,thick] (p0) -- (t02);

\draw[-latex,thick] (p11) -- (t11);
\draw[-latex,thick] (p12) -- (t12);
\draw[-latex,thick] (p13) -- (t13);
\draw[-latex,thick] (p13) -- (t14);

\draw[-latex,thick] (p23) -- (t21);
\draw[-latex,thick] (p24) -- (t22);

\draw[-latex,thick] (t01) -- (p11);
\draw[-latex,thick] (t01) -- (p12);

\draw[-latex,thick] (t02) -- (p13);

\draw[-latex,thick] (t11) -- (p21);
\draw[-latex,thick] (t12) -- (p22);
\draw[-latex,thick] (t13) -- (p23);
\draw[-latex,thick] (t14) -- (p24);
\draw[-latex,thick] (t21) -- (p31);
\draw[-latex,thick] (t22) -- (p32);
 
\end{tikzpicture}
}	
\end{minipage}
\begin{minipage}{1.3cm}
\hspace{1.5cm}
\end{minipage}

\caption{Three example systems given as Petri nets.}
\label{fig:example-systems}

\end{figure}

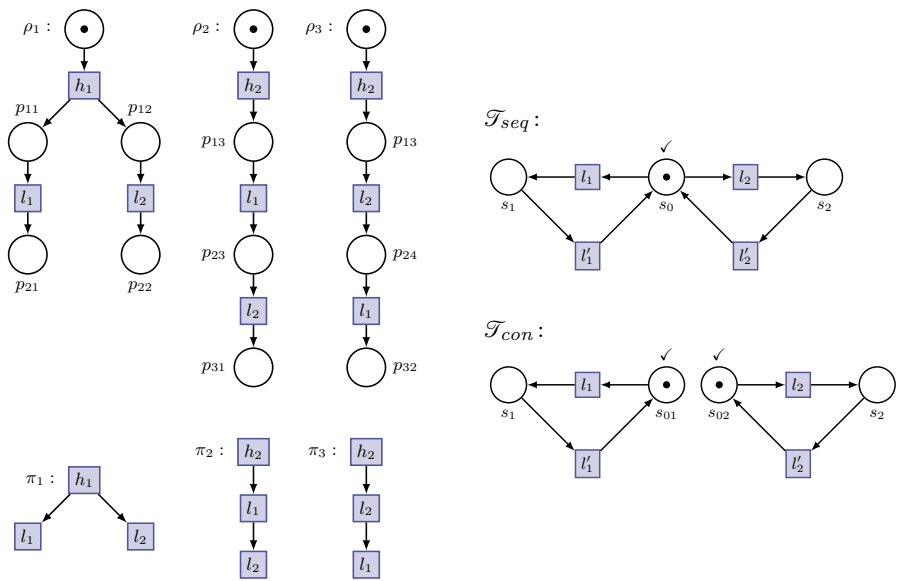
\begin{figure}[t]

\begin{minipage}{4.9cm}
\hspace{2mm}\scalebox{0.75}{
  \begin{tikzpicture}

\node[place,
    tokens=1,
    label=left: $\rho_1:\ $] (p01) at (1,6) {};
    
\node[place,
    tokens=1,
    label=left: $\rho_2:\ $] (p02) at (4,6) {};

\node[place,
    tokens=1,
    label=left: $\rho_3:\ $] (p03) at (6,6) {};
 
\node[place,
    label=$p_{11}$] (p11) at (0,4) {};

\node[place,
    label=$p_{12}$] (p12) at (2,4) {};
 
\node[place,
    label=left:$p_{13}$] (p132) at (4,4) {};
    
\node[place,
    label=right:$p_{13}$] (p133) at (6,4) {};
        
\node[place,
    label=below:$p_{21}$] (p21) at (0,2) {};
    
\node[place,
    label=below:$p_{22}$] (p22) at (2,2) {};
    
\node[place,
    label=left:$p_{23}$] (p23) at (4,2) {};

\node[place,
    label=right:$p_{24}$] (p24) at (6,2) {};

\node[place,
    label=left:$p_{31}$] (p31) at (4,0) {};
    
\node[place,
    label=right:$p_{32}$] (p32) at (6,0) {};


\node[transition,
    label=left: $\pi_1:\ $] (tr10) at (1,-2) {$h_1$};
\node[transition,
    label=] (tr11) at (0,-3) {$l_1$};
\node[transition,
    label=] (tr12) at (2,-3) {$l_2$};
\draw[-latex,thick] (tr10) -- (tr11);
\draw[-latex,thick] (tr10) -- (tr12);


\node[transition,
    label=left: $\pi_2:\ $] (tr20) at (4,-1.5) {$h_2$};
\node[transition,
    label=] (tr21) at (4,-2.5) {$l_1$};
\node[transition,
    label=] (tr22) at (4,-3.5) {$l_2$};
\draw[-latex,thick] (tr20) -- (tr21);
\draw[-latex,thick] (tr21) -- (tr22);


\node[transition,
    label=left: $\pi_3:\ $] (tr30) at (6,-1.5) {$h_2$};
\node[transition,
    label=] (tr31) at (6,-2.5) {$l_2$};
\node[transition,
    label=] (tr32) at (6,-3.5) {$l_1$};
\draw[-latex,thick] (tr30) -- (tr31);
\draw[-latex,thick] (tr31) -- (tr32);

\node[transition,
    label=] (t01) at (1,5) {$h_1$};

\node[transition,
    label=] (t022) at (4,5) {$h_2$};
    
\node[transition,
    label=] (t023) at (6,5) {$h_2$};

\node[transition,
    label=] (t11) at (0,3) {$l_1$};

\node[transition,
    label=] (t12) at (2,3) {$l_2$};
    
\node[transition,
    label=] (t13) at (4,3) {$l_1$};
 
\node[transition,
    label=] (t14) at (6,3) {$l_2$};
 
\node[transition,
    label=] (t21) at (4,1) {$l_2$};
 
\node[transition,
    label=] (t22) at (6,1) {$l_1$};

\draw[-latex,thick] (p01) -- (t01);
\draw[-latex,thick] (p02) -- (t022);
\draw[-latex,thick] (p03) -- (t023);

\draw[-latex,thick] (p11) -- (t11);
\draw[-latex,thick] (p12) -- (t12);
\draw[-latex,thick] (p132) -- (t13);
\draw[-latex,thick] (p133) -- (t14);

\draw[-latex,thick] (p23) -- (t21);
\draw[-latex,thick] (p24) -- (t22);

\draw[-latex,thick] (t01) -- (p11);
\draw[-latex,thick] (t01) -- (p12);

\draw[-latex,thick] (t022) -- (p132);
\draw[-latex,thick] (t023) -- (p133);

\draw[-latex,thick] (t11) -- (p21);
\draw[-latex,thick] (t12) -- (p22);
\draw[-latex,thick] (t13) -- (p23);
\draw[-latex,thick] (t14) -- (p24);
\draw[-latex,thick] (t21) -- (p31);
\draw[-latex,thick] (t22) -- (p32);
 
\end{tikzpicture}
}
\end{minipage}
\begin{minipage}{1.5cm}
\hspace{1.5cm}
\end{minipage}
\begin{minipage}{2.5cm}

 \calT{$\mathit{seq}$}:  

\scalebox{0.7}{
  \begin{tikzpicture}

\node[place,
    label=below:$s_1$] (p91) at (0,9) {};
    
\node[transition,
    label=] (t91) at (1.5,9) {$l_1$};
    
\node[place,
    tokens=1,
    label=above:$\tick$,label=below:$s_0$] (p92) at (3,9) {};
    
\node[transition,
    label=] (t92) at (4.5,9) {$l_2$};
    
\node[place,
    label=below:$s_2$] (p93) at (6,9) {};

\node[transition,
    label=] (t72) at (1.5,7.5) {$l'_1$};

\node[transition,
    label=] (t73) at (4.5,7.5) {$l'_2$};


\draw[-latex,thick] (p92) -- (t91);
\draw[-latex,thick] (p92) -- (t92);

\draw[-latex,thick] (p91) -- (t72);
\draw[-latex,thick] (p93) -- (t73);


\draw[-latex,thick] (t91) -- (p91);
\draw[-latex,thick] (t92) -- (p93);

\draw[-latex,thick] (t72) -- (p92);
\draw[-latex,thick] (t73) -- (p92);

\end{tikzpicture}
}

\vspace{5mm}

 \calT{$\mathit{con}$}:   

\scalebox{0.7}{
  \begin{tikzpicture}

\node[place,
    label=below:$s_1$] (p91) at (0,9) {};
    
\node[transition,
    label=] (t91) at (1.5,9) {$l_1$};
    
\node[place,
    tokens=1,
    label=above:$\tick$,label=below:$s_{01}$] (p921) at (3,9) {};
    
\node[place,
    tokens=1,
    label=above:$\tick$,label=below:$s_{02}$] (p922) at (4,9) {};
    
\node[transition,
    label=] (t92) at (5.5,9) {$l_2$};
    
\node[place,
    label=below:$s_2$] (p93) at (7,9) {};

\node[transition,
    label=] (t72) at (1.5,7.5) {$l'_1$};

\node[transition,
    label=] (t73) at (5.5,7.5) {$l'_2$};


\draw[-latex,thick] (p921) -- (t91);
\draw[-latex,thick] (p922) -- (t92);

\draw[-latex,thick] (p91) -- (t72);
\draw[-latex,thick] (p93) -- (t73);


\draw[-latex,thick] (t91) -- (p91);
\draw[-latex,thick] (t92) -- (p93);

\draw[-latex,thick] (t72) -- (p921);
\draw[-latex,thick] (t73) -- (p922);

\end{tikzpicture}

}	
\end{minipage}

\caption{
  \emph{Left}: The three maximal runs $\rho_1, \rho_2$ and $\rho_3$ of $\mathcal N_C$ from Fig.~\ref{fig:example-systems}, resulting by resolving every nondeterministic choice in $\mathcal N_C$,
  and their corresponding concurrent traces $\pi_1, \pi_2$ and $\pi_3$.
  \emph{Right}: A sequential test \calT{$\mathit{seq}$} for the concurrent hyperproperty
  that every pair of concurrent traces $\pi$ and $\pi'$ must agree on the occurrence \emph{and} sequential ordering of the low-security events $l_1$ and $l_2$. 
  In the test, the events $l_1$ and $l_2$ refer to $\pi$ and $l'_1$ and $l'_2$ to $\pi'$. The place marked with the symbol $\tick$ notifies
  a successful test. Below is a  concurrent test \calT{$\mathit{con}$} for the weaker concurrent hyperproperty
  that every pair of concurrent traces $\pi$ and $\pi'$ must agree on the occurrence of the low-security events $l_1$ and $l_2$,
  but not on their sequential ordering.
  For instance, each each $l_1$ must be matched by $l'_1$ before the next $l_1$ can occur, but $l_2$ may occur in between $l_1$ and $l'_1$.
}
\label{fig:example-test}

\end{figure}

For a start, consider the system $\mathcal N_A$ shown on the left in
Fig.~\ref{fig:example-systems}. We are interested in the secrecy property that the
system's low-security behavior, as observable in the low-security
events $l_1$ and $l_2$, is not affected by the
high-security events $h_1$ and $h_2$.  Our system is secure. This is captured by
the hyperproperty that \emph{all} traces must agree on the occurrences and the ordering 
of $l_1$ and $l_2$, and indeed, the system has only two traces, $h_1 \cdot l_1
\cdot l_2$ and $h_2 \cdot l_1 \cdot l_2$, which, when projected to
$\{l_1, l_2\}$, both result in the same sequence $l_1 \cdot l_2$ of
low-security events.

Next, consider system $\mathcal N_B$ shown in the middle in
Fig.~\ref{fig:example-systems}.  
Informally, the system is still secure in the sense that an observer who sees only $l_1$ and $l_2$ cannot distinguish the situation where $h_1$ has occurred from the situation where $h_2$ has occurred.
However, our previous hyperproperty is violated.  The system has four traces: 
$h_1 \cdot l_1$, $h_1 \cdot l_2$, $h_2\cdot l_1$, and $h_2\cdot l_2$,
which, when projected to $\{l_1, l_2\}$, result in two different
traces, $l_1$ and $l_2$. This issue is due to the nondeterministic choice between $l_1$ and $l_2$, and can be addressed with possibilistic information-flow
properties like \emph{generalized noninterference}~\cite{McCullough:1987:GNI}. Generalized
noninterference is weaker than normal noninterference: it requires that for every pair of traces $\pi, \pi'$
\emph{there exists} another trace $\pi''$, such that (1) $\pi''$
agrees with $\pi$ on the low-security events $\{l_1, l_2\}$ and (2)
$\pi''$ agrees with $\pi'$ on the high-security events $\{h_1,
h_2\}$. Generalized noninterference is satisfied in $\mathcal N_B$. For
example, for $\pi=h_1 \cdot l_1$ and $\pi'=h_2\cdot l_2$, there exists
$\pi''= h_2\cdot l_1$, which agrees with $\pi$ on $\{l_1, l_2\}$ and
with $\pi'$ on $\{h_1, h_2\}$.

Finally, consider the \emph{concurrent} system $\mathcal N_C$ shown on the right in Fig.~\ref{fig:example-systems}.
With the interpretation of concurrency as nondeterministic
interleaving, the system has the four traces $h_1\cdot l_1 \cdot l_2$,
$h_1\cdot l_2 \cdot l_1$, $h_2\cdot l_1 \cdot l_2$, and $h_2\cdot l_2
\cdot l_1$. Generalized noninterference is satisfied. However, the system is clearly not secure, because $h_1$
causes concurrent behavior, while $h_2$ causes sequential behavior. In
a concurrent setting, this difference could be recognized by an
attacker, who might, for example, synchronize with the system on a
particular ordering, such as $l_1 \cdot l_2$. In a trace that begins
with $h_1$, this will always work, while in traces that begin with
$h_2$, the attacker might observe a deadlock when the system performs
the order $l_2 \cdot l_1$.

In the security literature, this phenomenon has lead to the study of
\emph{branching-time} information-flow properties based on various
notions of (bi-)simulation (cf.~\cite{busi_gorrieri_2009}). Often, however, such equivalences are too
fine-grained, because they expose the point in time when an internal
decision is made. Linear-time properties, and, hence, hyperproperties
abstract from such implementation details. Can hyperproperties
nevertheless recognize the difference between concurrent and
sequential behavior?

 In this paper, we propose \emph{concurrent hyperproperties} as a
 positive answer to this question. Hyperproperties are based on the
 partial-order interpretation of concurrency. We stick to Clarkson and
 Schneider's definition of hyperproperties as sets of sets of traces,
 but generalize traces to \emph{concurrent} traces, which we define as
 partially ordered multisets (pomsets). 
 Figure~\ref{fig:example-test} shows the three maximal
 runs $\rho_1, \rho_2$ and $\rho_3$ of system $\mathcal N_C$ 
 and their corresponding concurrent traces. 
 In a run, every nondeterministic choice has been resolved, but concurrent executions
 remain visible, like the concurrency of the transitions labeled with $l_1$ and $l_2$ in
 $\rho_1$.
 The concurrency of run $\rho_1$ is reflected in the partial order of the concurrent trace~$\pi_1$.
 Note that $\mathcal N_C$ has four traces under the interleaving semantics
 (corresponding to the two nondeterministic choices and the two
 possible interleavings) but only three concurrent traces, because the concurrent execution is not
 resolved by nondeterminism. 
 Since the concurrency is still present in the concurrent traces, 
 a concurrent hyperproperty can distinguish nondeterminism from concurrency.
 Continuing our example, we can now specify secrecy in concurrent systems
 like $\mathcal N_C$  as the concurrent
 hyperproperty where every pair of concurrent traces agrees on the
 occurrence and ordering of the low-security events. Our example system clearly violates this requirement.
 
 In the paper, we give a formal definition of concurrent
 hyperproperties and then provide an explicit mechanism for describing
 concurrent hyperproperties. We base this mechanism the concept of
 \emph{testing processes} due to DeNicola and Hennessy
 \cite{DH84,Hen88}.  There the interaction of a (nondeterministic)
 process and a user is explicitly formalized using a synchronous
 parallel composition.  The user is formalized by a \emph{test}, which
 is a process with some states marked as a \emph{success}.  
 It is defined when a process \emph{may} pass a test and 
 when it \emph{must} pass a test.
 We transfer the concept of testing to concurrent traces. A concurrent hyperproperty
 is given as a test that has interactions with multiple concurrent
 runs. The test is successful for a given set of concurrent traces
 if it succeeds for all combinations of concurrent traces from the set.

 For our example, such a test \calT{$\mathit{seq}$} is
 shown on the right in Fig.~\ref{fig:example-test}. 
 It can interact with any two of the runs $\rho_1, \rho_2, \rho_3$
 corresponding to any two of the traces $\pi_1, \pi_2, \pi_3$ of \calN{$C$}.
 The interaction is via parallel composition that synchronizes on all transitions
 with the same label. To this end, the first run under test keeps the original labels
 $l_1$ and $l_2$, whereas the second run uses primed copies $l_1'$ and $l_2'$
 of these labels. 
 Thus \calT{$\mathit{seq}$} allows for both
 possible orderings ($l_1$ then $l_2$, and $l_2$ then $l_1$) in the
 first trace and enforces that the second trace exhibits the same order.  
 When \calT{$\mathit{seq}$} is applied to the runs of the 
 concurrent system \calN{$C$} shown on the left 
 of Fig.~\ref{fig:example-test}, it turns out that they may not pass this test,
 for instance, when $\rho_1$ and $\rho_3'$, i.e., $\rho_3$ with primed labels,
 are tested for the sequence $l_1 \cdot l_1' \cdot l_2 \cdot l_2'$, 
 this leads to a deadlock after $l_1$.
 This shows that the concurrent system \calN{$C$} does not satisfy the
 concurrent hyperproperty. We will examine this in more detail in Section~\ref{sec:examples}.
 
 The test \calT{$\mathit{con}$} checks a weaker concurrent hyperproperty,
 namely that each occurrence of $l_1$ is matched by an occurrence of $l'_1$
 before the next occurrence of $l_1$, and similarly for $l_2$ and $l'_2$,
 but $l_2$ may occur in between $l_1$ and $l'_1$.
 When \calT{$\mathit{con}$} is applied to any two of the runs $\rho_1, \rho_2, \rho_3$
 shown on the left of Fig.~\ref{fig:example-test}, it turns out that they must pass this test.
 This shows that the concurrent system \calN{$C$} satisfies this weaker
 concurrent hyperproperty. For more details see Section~\ref{sec:examples}.

Our paper is organized as follows.
In Section~\ref{sec:conc-hyper} we define the notion of concurrent hyperproperties
and give examples of ascending sophistication.
In Section~\ref{sec:nets} we recall the basic concepts from Petri nets
that we take as our semantic model of concurrent systems. In particular,
we define concurrent runs and the parallel composition of nets.
In Section~\ref{sec:testing} we adapt the concept of testing developed
by DeNicola and Hennessy to the setting of Petri nets.
In Section~\ref{sec:examples} we discuss how various examples of
concurrent hyperproperties can be tested.
In Section~\ref{sec:decidability} we briefly discuss the decidability
of universal must testing and establish an undecidability result for universal may testing.
In Section~\ref{sec:conclusion} we conclude the paper.

\paragraph{Dedication.} 
We dedicate our paper to Jifeng He on the occasion of his 80th birthday.
Jifeng has made many contributions to 
formalizing and relating different semantic models of computing, as exemplified
in his book `Unifying Theories of Programming' with Tony Hoare~\cite{HoaHe98}.
Out of this work grew also Jifeng's interest in testing
\cite{AichHe07,SuFPHeS15,SuWMPHeCS17}, the concept that we employ for hyperproperties
in this paper, although in an abstract setting of testing processes
as introduced by DeNicola and Hennessy.
The second author has very pleasant memories of the close cooperation with Jifeng
within the 
EU Basic Research Action  ProCoS (Provably Correct Systems) during the period 1989--1995~\cite{ProCoS94},
and of various scientific meetings, in particular in Oxford, Oldenburg, and Shanghai.

\section{Concurrent Hyperproperties}
\label{sec:conc-hyper}

Clarkson and Schneider defined \emph{hyperproperties} as a generalization of trace properties, which are sets of traces, to sets of sets of traces~\cite{ClaSch10}. To give an analogous definition of \emph{concurrent} hyperproperties, we generalize traces to \emph{concurrent} traces, which we define as partially ordered multisets (pomsets). 

Let $\Sigma$ be a set of labels. A $\Sigma$-labeled partially ordered set is a triple $(X, <, \ell)$ where $<$ is an irreflexive partial order on a set $X$ and $\ell: X \rightarrow \Sigma$ is a labeling function. Two such sets $(X,<,\ell)$ and $(X', <', \ell')$ are \emph{isomorphic} if there exists a bijective mapping $f: X \rightarrow X'$ such that $f(x) < f(y) \Leftrightarrow x<y$ and $\ell'(f(x)) = \ell(x)$.
A \emph{partially ordered multiset (pomset)} over $\Sigma$ is an isomorphy class of $\Sigma$-labeled partial ordered sets,
denoted as $[(X, <, \ell)]$.
A \emph{totally ordered multiset (tomset)} is a pomset where $<$ is a total order \cite{Pratt84}.

We then refer to tomsets over $\Sigma$ as \emph{traces} and pomsets over $\Sigma$ as \emph{concurrent traces}. A \emph{trace property} is a set of traces; a \emph{hyperproperty} is a set of sets of traces. Analogously, a \emph{concurrent trace property} is a set of concurrent traces, and a set of sets of concurrent traces is a \emph{concurrent hyperproperty}. We denote with $\mathbb{T}(\Sigma)$ the set of all concurrent traces over $\Sigma$.

\begin{example} 
\label{ex:conc-trace-prop}
A simple information flow policy for a concurrent system is to forbid any dependency of a low-security event labeled $l$ (for \emph{low}) 
on a high-security event labeled $h$ (for \emph{high}). 
Let $\Sigma=\{ l, h \}$. The policy can be expressed as the concurrent trace property 
$$T_1 = \{\ [(X,<,\ell)] \in \mathbb{T}(\Sigma)\ \mid\ \forall x, y \in X. x<y \Rightarrow \ell(x) \neq h \vee \ell(y) \neq l \}.$$
\end{example}

\begin{example} 
\label{ex:simple-conc-hyper-prop}
Consider the hyperproperty that every pair of concurrent traces agrees on the occurrence of the low-security events, independent
on any other event.
Let $\Sigma_{\mathit{low}}$ be the set of low-security events. The requirement can then be formalized as the following concurrent hyperproperty $H_1$: 
\[
\begin{array}{llll}
     H_1 = \{\ T \subseteq \mathbb{T}(\Sigma)\ \mid & \forall\, [(X,<,\ell)], [(X',<',\ell')] \in T.\\[1mm]
       &\mbox{\ \ $\exists$ bijection } f: X_{\mathit{low}} \rightarrow X'_{\mathit{low}}.  \forall x\in X_{\mathit{low}}.\,\ell'(f(x)) = \ell(x)\, \} 
\end{array}
\]
where
$X_{\mathit{low}} = \{ x \in X \mid \ell(x) \in \Sigma_{\mathit{low}}\}$
and
$X'_{\mathit{low}} = \{ x \in X' \mid \ell'(x) \in \Sigma_{\mathit{low}}\}$.

In the introduction, we discussed the concurrent hyperproperty that every pair of concurrent traces agrees both on the occurrence and the ordering of the low-security events. This requirement can be formalized as the following concurrent hyperproperty $H_2$: 
\[
\begin{array}{llll}
     H_2 = \{\ T \subseteq \mathbb{T}(\Sigma)\ \mid & \forall\, [(X,<,\ell)], [(X',<',\ell')] \in T.\\[1mm]
       &\mbox{\quad $\exists$ bijection } f: X_{\mathit{low}} \rightarrow X'_{\mathit{low}}. \\[1mm]
       &\mbox{\qquad }   (\ \ \forall x\in X_{\mathit{low}}.\,\ell'(f(x)) = \ell(x) \\[1mm]
       &\mbox{\qquad } \land \forall x,y \in X_{\mathit{low}}.\,f(x) <' f(y) \Leftrightarrow x<y\,)\, \} 
\end{array}
\]

\end{example}

\begin{example}  
    \label{ex:conc-hyper-prop-gni}
    As a final example, we adapt the notion of generalized noninterference (GNI)~\cite{McCullough:1987:GNI} to concurrent traces. 
    We identify the events as low-security and high-security: $\Sigma = \Sigma_{\mathit{low}} \cup \Sigma_{\mathit{high}}$. The policy then requires that for every pair of concurrent traces there exists a third concurrent trace that agrees with the first trace on the low-security events and with the second trace on the high-security events. Unlike the trace-based version discussed in the introduction, this version of GNI distinguishes nondeterminism from concurrency; in the example system $\mathcal N_C$ shown on the right in Fig.~\ref{fig:example-systems}, GNI on traces is satisfied, but GNI on concurrent traces is violated.
GNI on concurrent traces is expressed by the following concurrent hyperproperty $H_3$:
  \[
\begin{array}{rlll}
   H_3 = \{\ T \subseteq \mathbb{T}(\Sigma)\ \mid & 
           \forall\, [(X,<,\ell)], [(X',<',\ell')] \in T. \\ & \qquad  \exists [(X'',<'',\ell'')] \in T.\
                      F_{\mathit{low}} \wedge G_{\mathit{high}} \} \\[1mm]
\end{array}
\]
where 
\[
\begin{array}{rlll}
F_{\mathit{low}}\ \equiv
            &\ \exists\, \mbox{bijection } f: X_{\mathit{low}} \rightarrow X''_{\mathit{low}}. \\[1mm]
            &\mbox{\qquad }   (\ \ \forall x\in X_{\mathit{low}}.\,\ell''(f(x)) = \ell(x) \\[1mm]
            &\mbox{\qquad } \land \forall x,y \in X_{\mathit{low}}.\, f(x) <'' f(y) \Leftrightarrow x<y\,), \\[2mm]
G_{\mathit{high}}\ \equiv        
            &\  \exists\,  \mbox{bijection }  g: X'_{\mathit{high}} \rightarrow X''_{\mathit{high}}. \\[1mm]
            &\mbox{\qquad }   (\ \ \forall x\in X'_{\mathit{high}}.\,\ell''(g(x)) = \ell'(x) \\[1mm]
            &\mbox{\qquad } \land \forall x,y \in X'_{\mathit{high}}.\, g(x) <'' g(y) \Leftrightarrow x<'y\,), \\[2mm]
X_{\mathit{low}} = 
            & \{ x \in X \mid \ell(x) \in \Sigma_{\mathit{low}}\}, \\[2mm] 
X''_{\mathit{low}} = 
            & \{ x \in X'' \mid \ell''(x) \in \Sigma_{\mathit{low}}\}, \\[2mm]
X'_{\mathit{high}} = 
            & \{ x \in X' \mid \ell'(x) \in \Sigma_{\mathit{high}} \}, \\[2mm]
X''_{\mathit{high}} = 
            & \{ x \in X'' \mid \ell''(x) \in \Sigma_{\mathit{high}} \}.
\end{array}
\]

\end{example}

\section{Petri Nets}
\label{sec:nets}

As a model for concurrent systems we take Petri nets because they distinguish
the fundamental concepts of causal dependency, nondeterministic choice, and concurrency
explicitly.
We consider here safe Petri nets  \cite{Rei85}, with the transitions labeled by actions
which serve as synchronization points in a parallel composition
of such nets. We use the notation from \cite{Old91}, which is inspired by
\cite{Gol88b}.
{\it A Petri net} or simply {\it net} is a structure
  $\calN{} = ( \rmA{}, \rmPl{}, \arrow, \rmM{0})$,
where
  \begin{enumerate}
    \item \rmA{} is a finite communication alphabet with $\tau \not\in A$,
    \item \rmPl{} is a possibly infinite set of {\it places},
    \item \arrow\ $\subseteq$ $\calP{$\mathit{nf}$} (\rmPl{})$ $\times$ (\rmA{} $\cup$ \{ $\tau$ \}) $\times$ $\calP{$\mathit{nf}$} (\rmPl{})$ 
          is the {\it transition relation},
    \item \rmM{0} $\in$ $\calP{$\mathit{nf}$} (\rmPl{})$ is the {\it initial marking}.  
  \end{enumerate}

\noindent
We let \rmp{}, \rmq{}, \rmr{} range over \rmPl{}. The notation
$\calP{$\mathit{nf}$} (\rmPl{})$ stands for the set of all non-empty, finite subsets of \rmPl{}. An
element (\rmI{}, \rmu{}, \rmO{}) $\in$ \arrow\ with $\rmI{}, \rmO{} \in \calP{$nf$}(\rmPl{})$
and $\rmu{} \in \rmA{} \cup \{\tau\}$
is called a {\it transition} ({\it labeled with the action \rmu{}}) and written as
\[
 \rmI{}\ \arrowu{\rmu{}}\ \rmO{}.\sindex{\rmI{}\ \arrowu{\rmu{}}\ \rmO{}}
\]
For a transition \rmt{} = \rmI{} \arrowu{\rmu{}} \rmO{} its {\it preset} or {\it input} is given by \pre{}(\rmt{}) = \rmI{},
its {\it postset} or {\it output} by \post{}(\rmt{}) = \rmO{}, and its action by \act{}(\rmt{}) = \rmu{}.
The letter $\tau$ is intended to model an \emph{internal} action.

In the graphical representation of a net \calN{} = (\rmA{}, \rmPl{}, \arrow, \rmM{0}) 
we mention the alphabet \rmA{} separately 
and display the components \rmPl{}, \arrow\ and \rmM{0} as usual. 
Places \rmp{} $\in$ \rmPl{} are
represented as circles $\bigcirc$ with the name \rmp{} outside and transitions
\[
 \rmt{} = \{\rmp{1}, \ldots, \rmp{$m$}\}\ \arrowu{\rmu{}}\ \{\rmq{1}, \ldots, \rmq{$n$}\}
\]
as boxes \fbox{\rmu{}}\ carrying the label \rmu{} inside and connected via directed arcs to the
places in \pre{(\rmt{})} and \post{(\rmt{})}:

\vspace{1mm}

\begin{center}
\scalebox{0.9}{
  \begin{tikzpicture}

\node[place,
    label=left:$p_1$] (p13) at (1,3) {};
    
\node[align=center] (d3) at (2.1,3) { $\cdots$ };
    
\node[place,
    label=right:$p_m$] (p33) at (3,3) {};
    
\node[transition,
    label=] (t) at (2,2) {$u$};
 
\node[place,
    label=left:$q_1$] (p11) at (1,1) {};
    
\node[align=center] (d1) at (2.1,1) { $\cdots$ };
    
\node[place,
    label=right:$q_n$] (p31) at (3,1) {};
 
\draw[-latex,thick] (p13) -- (t);
\draw[-latex,thick] (p33) -- (t);
 
\draw[-latex,thick] (t) -- (p11);
\draw[-latex,thick] (t) -- (p31);

\end{tikzpicture}
}
\end{center}

\noindent
Since \pre{(\rmt{})} and \post{(\rmt{})} need not be disjoint, some of the outgoing arcs of \fbox{\rmu{}}\ 
may actually point back to places in \pre{(\rmt{})} and thus introduce {\it cycles}. 
Graphically, we employ then double-headed arrows between \fbox{\rmu{}} and the places in $\pre{(\rmt{})} \cap \post{(\rmt{})}$.
The initial marking \rmM{0} is represented by putting a token $\bullet$ into the circle of each \rmp{} $\in$ \rmM{0}.

Starting from the initial marking, the firing of transitions creates new markings $M \in \calP{$\mathit{nf}$} (\rmPl{})$,
which represent the global states of a Petri net.
Formally, a transition $t$ is \emph{enabled} at a marking $M$ if $\mathit{pre}(t) \subseteq M$.
\emph{Firing} such a transition $t$ at $M$ 
yields the successor marking 
$M' = (M - \mathit{pre}(t)) \cup \mathit{post}(t)$. We write then $M \firable{t} M'$.
We assume here that $\cup$ is a disjoint union, which is satisfied 
if the net is \emph{contact-free}, i.e., if for all $t \in \tr$
and all reachable markings~$M$
\[
  \mathit{pre}(t) \subseteq M \Rightarrow \mathit{post}(t) \subseteq (\rmPl{} - M) \cup \mathit{pre}(t).
\] 
The set of \emph{reachable markings} of a net $\mathcal N$
is defined by
\[
 \reach(\mathcal{N}) = \{M \mid \exists n \in \mathbb{N}.\, \exists\ t_1, \ldots, t_n \in\, \transitions.\ 
M_0\firable{t_1}M_1 \firable{t_2} \ldots \firable{t_n} M_n = M \}. 
\]
For $n=0$ inside this set, it is understood that $M_0 = M$ holds, so $M_0 \in \reach(\mathcal{N})$.
In the present setting, all reachable markings are non-empty, finite sets of places.
Such Petri nets are called \emph{safe} or \emph{1-bounded} because
every reachable marking contains at most one token per place.
In general place/transition nets, the reachable markings can be multisets representing
multiple tokens per place.

\subsection{Causal Nets and Runs}

Concurrent computations of a net can be described by {\it causal nets} \cite{Pet77,Rei85}. Informally, a causal net
is an acyclic net where all choices have been resolved. It can be seen as a
net-theoretic way of defining a partial order among the occurrences of
transitions in a net to represent their causal dependency. 

We need more notation for a net \calN{}\ = (\rmA{}, \rmPl{}, \arrow, \rmM{0}). 
For a place \rmp{} $\in$ \rmPl{} its \emph{preset} is defined by
$\pre{(\rmp{})}\sindex{\pre{(\rmp{})}} = \{ \rmt{} \in \arrow\ |\ \rmp{} \in \post{(\rmt{})}\ \}$
and its \emph{postset} by
$\post{(\rmp{})}\sindex{\post{(\rmp{})}} = \{ \rmt{} \in \arrow\ |\ \rmp{} \in \pre{(\rmt{})}\ \}$.
The {\it flow relation} \calFN{}\ $\subseteq$ \rmPl{} $\times$ \rmPl{} on the places of \calN{}\ is given by
\[
  \rmp{}\ \calFN{}\ \rmq{}\sindex{\rmp{}\ \calFN{}\ \rmq{}}\ \mtext{if}\ 
  \exists\, \rmt{} \in \arrow{}\,.\ \rmp{} \in \pre{(\rmt{})} \mtext{and} \rmq{} \in \post{(\rmt{})}.
\]
\calFN{}\ is {\it well-founded} if there are no infinite backward chains
\[
 \cdots\ \rmp{3}\ \calFN{}\ \rmp{2}\ \calFN{}\ \rmp{1}.
\]

A {\it causal net} is a net \calN\ = ( \rmA{}, \rmPl{}, \arrow{}, \rmM{0}) such that
\begin{enumerate}
\item[(1)] all places are unbranched, i.e.,
           $\ \forall \rmp{}\ \in \rmPl{}\,.\ |\pre{(\rmp{})}| \leq 1 \mtext{and} |\post{(\rmp{})}| \leq 1,$
\item[(2)] the flow relation \calFN{}\ is well-founded, and
\item[(3)] the initial marking consists of all places without an ingoing arc, i.e.,
           \[
             \rmM{0}\ = \{\rmp{}\ \in \rmPl{}\ |\ \pre{(\rmp{})}\ = \emptyset \}.
           \]
\end{enumerate}

\noindent
By condition (1), there are no choices in \calN{}.
Condition (2) implies that the transitive closure of \calFN{} is irreflexive.
Thus a causal net \calN\ is acyclic, so each transition occurs only once.
Conditions (1)--(3) ensure that there are no
superfluous places and transitions in causal nets: every transition can fire
and every place is contained in some reachable marking.
Also, every causal net is safe.

Following Petri's intuition, causal nets should describe the concurrent computations
of a net. Thus we explain how causal nets relate to ordinary (safe)
nets. To this end, we use the following notion of embedding.

Let \calN{1}\ = (\rmA{1}, \rmPl{1}, \arrowd{1}, \rmM{01}) be a causal net and
  \calN{2}\ = (\rmA{2}, \rmPl{2}, \arrowd{2}, \rmM{02}) be a safe net,
where \rmM{01} and \rmM{02} denote the initial markings of \calN{1} and \calN{2}, respectively.
\calN{1}\ is a {\it causal net of} \calN{2} if \rmA{1}\ = \rmA{2}
and there exists a mapping $\mathf{}\ : \rmPl{1}\ \arrow\ \rmPl{2}$, 
which is extended elementwise to subsets $X \subseteq \rmPl{1}$ by putting
$f(X) = \{ f(p) \in \rmPl{2} \mid p \in X \}$, 
such that the following holds:
\begin{enumerate}
    \item \mathf{}(\rmM{01}) = \rmM{02},
    \item $\forall$ \rmM{} $\in$ \reach(\calN{1}).\ \mathf{}\, $\downarrow$ \rmM{}, the restriction of $f$ to $M \subseteq \rmPl{1}$, is injective,
    \item $\forall$ \rmt{} $\in$ \arrowd{1} .\ $ (\mathf{}(\pre{(\rmt{})}), \act{(\rmt{})}, \mathf{}(\post{(\rmt{})})) \in$ \arrowd{2} ,
\end{enumerate}
The mapping \mathf{}\ is called an {\it embedding of \calN{1}\ into \calN{2}}.
Note that $f$ distributes over the flow relation:
\[
 \forall\, p, q \in \rmPl{1}\, .\,( p\ \calF{$\calN{1}$}\ q\, \Rightarrow f(p)\ \calF{$\calN{2}$}\ f(q).
\]

In net theory, the pair $(\calN{1}, \mathf{})$ is called a {\it process} of $\calN{2}$ \cite{Pet77,BF88}.
We call it a (\emph{concurrent}) \emph{run} of $\calN{2}$ and  use the (possibly decorated) letter 
$\rho$ for runs. 
A run $\rho = (\calN{1}, \mathf{})$ of $\calN{2}$ is called \emph{maximal} if
\[
  \forall\, p \in \rmPl{1}\, .\,
  (\exists\,  q \in \rmPl{2}\, .\, f(p)\ \calF{$\calN{2}$}\ q  \Rightarrow
   \exists\, p' \in \rmPl{1}\, .\, p\ \calF{$\calN{1}$}\ p'),                
\]
so the run $\rho$ cannot stop at a place $p$ if there is an extension possible
at the corresponding place $f(p)$ in $\calN{2}$.

\subsection{Causal Nets Corresponding to Concurrent Traces}

A causal net \calN{}\ \emph{corresponds to} the concurrent trace (pomset) $[(X,<,\ell)]$, where
\begin{itemize}
\item $X = \arrow$, the set of transitions of \calN{} ,
\item $<$ is the transitive closure of the \emph{immediate causal successor} relation $<_m$
      between transitions: $\rmt{1} <_m \rmt{2}$ holds for $\rmt{1}, \rmt{2} \in \arrow\ $ if 
      $\ \post{(\rmt{1})} \cap \pre{(\rmt{2})} \neq \emptyset$,
\item $\ell(t) = \act{}(\rmt{})$ for every $t \in \arrow$.
\end{itemize}
The irreflexive partial order $\rmt{1} < \rmt{2}$\ expresses that transition \rmt{2}\ can
occur only after transition \rmt{1} has happened, so \rmt{2}\ {\it causally depends} on \rmt{1}. 
If for transitions \rmt{1}\ $\neq$ \rmt{2}\ neither $\rmt{1} < \rmt{2}$\ nor $\rmt{2} < \rmt{1}$ holds,
\rmt{1}\ and \rmt{2}\ are {\it causally independent} and can occur {\it concurrently}. 
Graphically, we represent these pomsets by showing each transition $t$ labeled with $\ell(t) = u$ 
as a box \fbox{\rmu{}} and connecting these boxes with arcs representing the 
immediate causal successor relation $<_m$ (see Fig.~\ref{fig:example-test}).

Also, vice versa, if a concurrent trace $[(X,<,\ell)]$ is given, it is easy to construct a causal
net \calN{}\ corresponding to the trace in the above sense. One just has to add the missing places
to turn the trace into a causal net.

\subsection{Parallel Composition}

Petri nets with disjoint sets of places, but possibly overlapping communication alphabets can be composed in parallel. 
Thereby transitions with different actions are performed asynchronously, whereas transitions 
with the \emph{same} action synchronize.
For \calN{$i$}  = (\rmA{$i$}, \rmPl{$i$}, \arrowd{$i$}, \rmM{0$i$}), $i$ = 1,2, with \rmPl{1}~$\cap$~\rmPl{2}~=~$\emptyset$
their \emph{parallel composition} is defined as follows:
\[ 
  \calN{1}\, \|\, \calN{2}  = (\rmA{1} \cup \rmA{2}, \rmPl{1} \cup \rmPl{2}, \arrow, \rmM{01} \cup \rmM{02}) ,
\]
where
\[
 \begin{array}{rccclc}
          \arrow = &      & \{ & (\rmI{},\rmu{},\rmO{}) \in \arrowd{1} \cup \arrowd{2}\ |\ \rmu{} \notin \rmA{1} \cap \rmA{2}           & \}
                              & \text{ (asynchrony)}\\[2mm]
                   & \cup\ & \{ & (\rmI{1} \cup \rmI{2}, \rma{}, \rmO{1} \cup \rmO{2})\ |\ \rma{} \in \rmA{1} \cap \rmA{2} \mtext{and}   & 
                              & \text{ (synchrony)}\\
                   &      &    & (\rmI{1},\rma{},\rmO{1}) \in \arrowd{1} \mtext{and} (\rmI{2},\rma{},\rmO{2}) \in \arrowd{2}            & \}.
 \end{array}
\]
Note that actions labeled with the internal action $\tau$ never synchronize because $\tau$ does not appear in any communication alphabet $A_i$.

Up to  bijective renaming of places, the parallel composition of nets
is commutative and associative, i.e., for all nets $\mathcal{N}_1, \mathcal{N}_2, \mathcal{N}_3$:
\begin{eqnarray*}
  \mathcal{N}_1 \ ||\ \mathcal{N}_2 &\ = \ & \mathcal{N}_2 \ ||\ \mathcal{N}_1, \\
  \mathcal{N}_1 \ ||\ (\mathcal{N}_2 \ ||\ \mathcal{N}_3) &\ = \ & (\mathcal{N}_1 \ ||\ \mathcal{N}_2) \ ||\ \mathcal{N}_3.
\end{eqnarray*}

\section{Testing}
\label{sec:testing}

The idea of \emph{testing} processes is due to De Nicola and Hennessy \cite{DH84,Hen88}.
There the interaction of a (nondeterministic) process and a
user is explicitly formalized using a synchronous parallel
composition.  The user is formalized by a \emph{test}, which is a
process with some states marked as a \emph{success}. 
The authors distinguish between two
options: a process may or must pass a test.  A process $P$ \emph{may}
pass a test $T$ if in \emph{some} maximal parallel computation with
$P$, synchronizing on transitions with the same label, the test $T$
reaches a \emph{success} state.  A process $P$ \emph{must} pass a test
$T$ if in \emph{all} such computations the test $T$ reaches a
\emph{success} state.

We transfer this notion of testing to Petri nets.
A \emph{test} is a Petri net, extended by a distinguished set $\tick \subseteq \rmPl{}$
of \emph{successful} places: $\calT{} = ( \rmA{}, \rmPl{}, \tick, \arrow, \rmM{0})$.
In the graphical notation, we mark each place of this subset by the symbol $\tick$.

To perform a test $\calT{}$ on a given Petri net $\calN{}$, we consider the
parallel composition $\calN{}\, \|\, \calT{}$.
A run $\rho = (\calN{$R$}, f)$ of $\calN{}\, \|\, \calT{}$ is \emph{deadlock free} if it is infinite,
and it \emph{terminates successfully} if it is finite and all places of \calT{} inside the parallel composition
without causal successor are marked with $\tick$.
A net $\calN{}$ \emph{may pass} a test  $\calT{}$ if 
there exists a maximal run of $\calN{}\, \|\, \calT{}$ which is deadlock free or terminates successfully.
A net $\calN{}$ \emph{must pass} a test  $\calT{}$ if 
all maximal runs of $\calN{}\, \|\, \calT{}$ are deadlock free or terminate successfully.

To \emph{check a hyperproperty} relating $k$ concurrent traces on a system represented
by a net $\calN{0}$, we investigate maximal runs $\rho_i = (\calN{$i$}, f_i)$ with 
$i = 1, \cdots, k$ of $\calN{0}$, where the causal nets \calN{$i$} correspond to
the concurrent traces of the hyperproperty, except that in \calN{$i$}
we relabel every action $u$ of \calN{$0$} into $u_i$.
We will test the parallel composition
$\calN{$1$}\, \|\, \cdots \, \|\, \calN{$k$}$.
The purpose of this relabeling is to have nets 
$\calN{$1$}, \dots, \calN{$k$}$ that do not synchronize in this composition.
To represent the hyperproperty, we suitably quantify 
existentially or universally over these $k$ runs of $\calN{0}$
and thus arrive at the following possibilities of testing:
\[
 \mathcal{Q_1}\, \rho_1, \cdots, \mathcal{Q_k}\, \rho_k.\ \calN{$1$}\, \|\, \cdots \, \|\, \calN{$k$}\ \mathcal{m} \mbox{ pass } \calT{},
\]
where $\mathcal{Q_i} \in \{\exists, \forall\}$ and $\mathcal{m} \in \{\mbox{may, must}\}$.
\calT{} uses the subscripted labels of the form $u_1, \dots , u_k$ to synchronize
with the actions in $\calN{$1$}, \dots, \calN{$k$}$.

We also use primed copies like $u'$ and $u''$ instead of subscripts.
For example, for $k=2$, we use one causal net \calN{} having the original actions of $\calN{0}$
and one causal $\mathcal{N}'$ with every action $u$ of \calN{$0$} relabled into a primed copy $u'$.
Then the above pattern specializes to
\[
 \mathcal{Q}\,  \rho .\, \mathcal{Q'}\, \rho'.\ \calN{}\, \|\, \mathcal{N}'\ \mathcal{m} \mbox{ pass } \calT{},
\]
where $\mathcal{Q}, \mathcal{Q'} \in \{\exists, \forall\}$ and $\mathcal{m} \in \{\mbox{may, must}\}$.
Whereas $\calN{}$ and $\mathcal{N}'$ have no common actions to synchronize on,
the test \calT{} will synchronize with $\calN{}$ and $\mathcal{N}'$ via common 
(unprimed and primed) actions, thereby checking the hyperproperty.
Note that the explicit quantifiers refer to runs of the system \calN{0} under test.
Once these runs are fixed, may and must corresponds to existential and universal quantification
over runs originating from the test.

\section{Examples}
\label{sec:examples}


\begin{figure}[t]

\begin{minipage}{4.9cm}

 $\calN{$1$}$:

 \scalebox{0.75}{
  \begin{tikzpicture}

\node[place,
    tokens=1,
    label=$p_1$] (p00) at (-1.5,2.5) {};
 
\node[place,
    label=below:$p_2$] (p01) at (-1.5,0.5) {};

\node[place,
    tokens=1,
    label=$q_0$] (p0) at (2,3) {};
 
\node[place,
    label=left:$q_1$] (p1) at (0,1.5) {};

\node[place,
    label=right:$q_2$] (p2) at (4,1.5) {};
    
\node[place,
    label=$q_3$] (p3) at (2,0) {};
    
\node[transition,
    label=] (t01) at (-1.5,1.5) {$l$};

\node[transition,
    label=left:] (t1) at (0,2.5) {$h_1$};

\node[transition,
    label=right:] (t2) at (4,2.5) {$h_2$};

\node[transition,
    label=] (t3) at (0,0.5) {$l_1$};

\node[transition,
    label=] (t4) at (4,0.5) {$l_1$};

\draw[-latex,thick] (p00) -- (t01);

\draw[-latex,thick] (p0) -- (t1);
\draw[-latex,thick] (p0) -- (t2);
\draw[-latex,thick] (p1) -- (t3);
\draw[-latex,thick] (p2) -- (t4);
 
\draw[-latex,thick] (t01) -- (p01);

\draw[-latex,thick] (t1) -- (p1);
\draw[-latex,thick] (t2) -- (p2);
\draw[-latex,thick] (t3) -- (p3);
\draw[-latex,thick] (t4) -- (p3);
 
\end{tikzpicture}
}	
\end{minipage}
\begin{minipage}{1.5cm}
\hspace{1.5cm}
\end{minipage}
\begin{minipage}{4.9cm}

$\calN{$2$}$:

\scalebox{0.75}{
  \begin{tikzpicture}

\node[place,
    tokens=1,
    label=$p_1$] (p00) at (-1.5,2.5) {};
 
\node[place,
    label=below:$p_2$] (p01) at (-1.5,0.5) {}; 

\node[place,
    tokens=1,
    label=$q_0$] (p0) at (2,3) {};
 
\node[place,
    label=left:$q_1$] (p1) at (0,1.5) {};

\node[place,
    label=right:$q_2$] (p2) at (4,1.5) {};
    
\node[place,
    label=$q_3$] (p3) at (2,0) {};
    
\node[transition,
    label=] (t01) at (-1.5,1.5) {$l$};

\node[transition,
    label=left:] (t1) at (0,2.5) {$h_1$};

\node[transition,
    label=right:] (t2) at (4,2.5) {$h_2$};

\node[transition,
    label=] (t3) at (0,0.5) {$l_1$};

\node[transition,
    label=] (t4) at (4,0.5) {$l_2$};

\draw[-latex,thick] (p00) -- (t01);

\draw[-latex,thick] (p0) -- (t1);
\draw[-latex,thick] (p0) -- (t2);
\draw[-latex,thick] (p1) -- (t3);
\draw[-latex,thick] (p2) -- (t4);
 
\draw[-latex,thick] (t01) -- (p01);

\draw[-latex,thick] (t1) -- (p1);
\draw[-latex,thick] (t2) -- (p2);
\draw[-latex,thick] (t3) -- (p3);
\draw[-latex,thick] (t4) -- (p3);
 
\end{tikzpicture}
}

\end{minipage}

\caption{\emph{Left}: Petri net $\calN{$1$}$ consists of two concurrent subnets, one performs only the low-security action $l$
and the other has a choice starting with different high-security actions $h_1$ and $h_2$, 
but then performing the \emph{same} low-security action $l_1$, no matter whether $h_1$ or $h_2$ was chosen.
\emph{Right}: Petri net $\calN{$2$}$ looks identical to  $\calN{$1$}$, but there is a subtle difference:
the subnet on the right-hand side performs either $l_1$ or $l_2$ depending
on the previous choice of $h_1$ or $h_2$, respectively.}
\label{fig:example-1-2}

\end{figure}
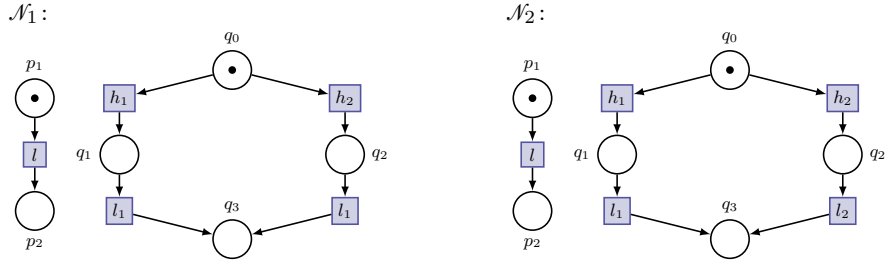


We examine concurrent trace properties 
and concurrent hyperproperties for examples of concurrent systems.
First consider the two Petri nets shown in Fig.~\ref{fig:example-1-2}.
The net $\calN{$1$}$ consists of two concurrent subnets, one performs the low-security action $l$
and the other has a choice starting with different high-security actions $h_1$ and $h_2$, 
but then both branches perform the same low-security action $l_1$.
The net $\calN{$2$}$ has the same structure, except that the choice in the
subnet on the right-hand side is now between performing action $l_1$ or action $l_2$ depending
on the previous choice of the high-security actions $h_1$ or $h_2$, respectively.
Note that due to the choices, each of the nets $\calN{$1$}$ and $\calN{$2$}$ have
two maximal runs, one with actions $h_1$ and one with action $h_2$.

Let us check the trace property whether the low-security action $l_1$ can occur after $l$,
independent of the high-security actions $h_1$ and $h_2$,
To this end, we use the following test~$\calT{}$: 

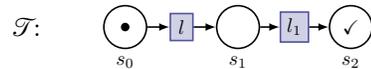
\begin{center}

\begin{minipage}{1cm}
$\calT{}$: 
\end{minipage}
\begin{minipage}{5cm}
\vspace{2mm}
\scalebox{0.75}{
  \begin{tikzpicture}

\node[place,
    tokens=1,
    label=below:$s_0$] (p0) at (0,0) {};
    
\node[transition,
    label=] (t1) at (1,0) {$l$};
 
\node[place,
    label=below:$s_1$] (p1) at (2,0) {}; 
    
\node[transition,
    label=] (t2) at (3,0) {$l_1$};

\node[place,
    label=below:$s_2$] (p2) at (4,0) {$\tick$};

\draw[-latex,thick] (p0) -- (t1);
\draw[-latex,thick] (p1) -- (t2);

\draw[-latex,thick] (t1) -- (p1);
\draw[-latex,thick] (t2) -- (p2);

\end{tikzpicture} .
}
\end{minipage}
 \end{center}

\noindent
This test is applied to each run of  $\calN{$1$}$  and  $\calN{$2$}$, respectively. We have
\[
 \forall \rho.\ \calN{$1,\rho$} \text{ must pass } \calT{},
\]
because \calT{} terminates successfully for each of the two maximal runs,
independent of the choice of $h_1$ or $h_2$.
Here \calN{$1,\rho$} denotes the net of the run $\rho$ of \calN{$1$}.

For $\calN{$2$}$ the test $\calT{}$ is less successful.
Let \calN{${2,h_1}$} and \calN{${2,h_2}$} be the nets for the two maximal runs of \calN{$2$},
depending on whether $h_1$ or $h_2$ is initially chosen. Then the parallel
composition with \calT{} yields the results shown in Fig.~\ref{fig:test-example-2}.
Note that synchronization is enforced on the common actions $l$ and $l_1$, whereas
$h_1$ and $h_2$ can occur asynchronously.
In $\calN{${2,h_1}$}\ \|\ \calT{}$, the test terminates successfully, whereas
$\calN{${2,h_2}$}\ \|\ \calT{}$ ends in a deadlock.
Thus  
\[
 \forall \rho.\ \calN{$2,\rho$} \text{ may pass } \calT{},
\]
but it is not the case that
$\forall \rho.\ \calN{$2,\rho$}$ must pass  $\calT{}$. Here \calN{$2,\rho$} denotes
the net of the run $\rho$ of \calN{$2$}.

%
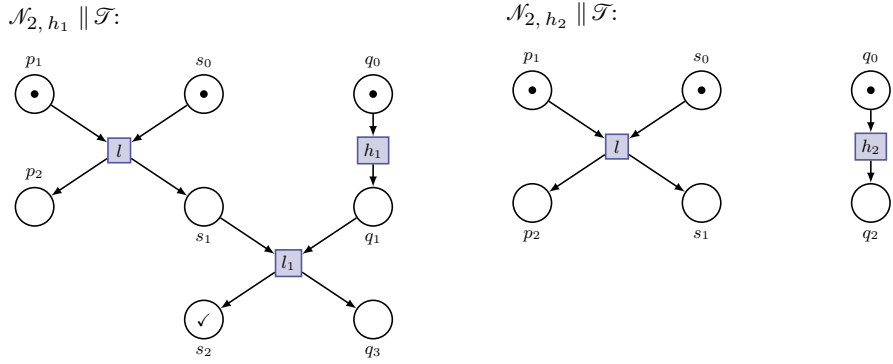
\begin{figure}[h]

\begin{minipage}{4.9cm}

 $\calN{${2,h_1}$}\, \|\, \calT{}$:
 
 \vspace{2mm}

 \scalebox{0.75}{
  \begin{tikzpicture}

\node[place,
    tokens=1,
    label=$p_1$] (p41) at (0,4) {};
 
\node[place,
    tokens=1,
    label=$s_0$] (p42) at (3,4) {};
 
\node[place,
    tokens=1,
    label=above:$q_0$] (p43) at (6,4) {};
    
\node[transition,
    label=] (t31) at (1.5,3) {$l$};
 
\node[transition,
    label=] (t32) at (6,3) {$h_1$};
    
\node[place,
    label=$p_2$] (p21) at (0,2) {};
    
\node[place,
    label=below:$s_1$] (p22) at (3,2) {};
    
\node[place,
    label=below:$q_1$] (p23) at (6,2) {};
    
\node[transition,
    label=] (t11) at (4.5,1) {$l_1$};

\node[place,
    label=below:$s_2$] (p01) at (3,0) {$\tick$};
    
\node[place,
    label=below:$q_3$] (p02) at (6,0) {};

\draw[-latex,thick] (p41) -- (t31);
\draw[-latex,thick] (p42) -- (t31);
\draw[-latex,thick] (p43) -- (t32);

\draw[-latex,thick] (p22) -- (t11);
\draw[-latex,thick] (p23) -- (t11);
 
\draw[-latex,thick] (t31) -- (p21);
\draw[-latex,thick] (t31) -- (p22);
\draw[-latex,thick] (t32) -- (p23);

\draw[-latex,thick] (t11) -- (p01);
\draw[-latex,thick] (t11) -- (p02);
 
\end{tikzpicture}
}	
\end{minipage}
\begin{minipage}{1.5cm}
\hspace{1.5cm}
\end{minipage}
\begin{minipage}{4.9cm}

\vspace{-16mm}

$\calN{${2,h_2}$}\, \|\, \calT{}$:

 \vspace{2mm}

\scalebox{0.75}{
  \begin{tikzpicture}

\node[place,
    tokens=1,
    label=$p_1$] (p41) at (0,4) {};
 
\node[place,
    tokens=1,
    label=$s_0$] (p42) at (3,4) {};
 
\node[place,
    tokens=1,
    label=above:$q_0$] (p43) at (6,4) {};
    
\node[transition,
    label=] (t31) at (1.5,3) {$l$};
 
\node[transition,
    label=] (t32) at (6,3) {$h_2$};
    
\node[place,
    label=below:$p_2$] (p21) at (0,2) {};
    
\node[place,
    label=below:$s_1$] (p22) at (3,2) {};
    
\node[place,
    label=below:$q_2$] (p23) at (6,2) {};

\draw[-latex,thick] (p41) -- (t31);
\draw[-latex,thick] (p42) -- (t31);
\draw[-latex,thick] (p43) -- (t32);

\draw[-latex,thick] (t31) -- (p21);
\draw[-latex,thick] (t31) -- (p22);
\draw[-latex,thick] (t32) -- (p23);

\end{tikzpicture}
}

\end{minipage}

\caption{Testing the two maximal runs of \calN{$2$}. In the middle, the places $s_0, s_1, s_2$ of test \calT{} in the parallel composition 
with these two runs are shown.
\emph{Left}: In $\calN{${2,h_1}$}\, \|\, \calT{}$, the test terminates successfully in $s_2$.
\emph{Right}: However, $\calN{${2,h_2}$}\, \|\, \calT{}$ ends in a deadlock, i.e., in places without $\tick$.}

\label{fig:test-example-2}

\end{figure}
%

\subsection{Testing the Concurrent Hyperproperties $H_1$ and $H_2$}

Next we turn to Section~\ref{sec:introduction} and consider the three runs shown in Fig.~\ref{fig:example-test}
stemming from system \calN{$C$} in Fig.~\ref{fig:example-systems}. 
First we check with the sequential test \calT{$\mathit{seq}$} of Fig.~\ref{fig:example-test} 
the concurrent hyperproperty whether every pair of concurrent traces $\pi$ and $\pi'$ 
agrees on the occurrence and ordering of the low-security events $l_1$ and~$l_2$. 
This is property $H_2$ in Example~\ref{ex:simple-conc-hyper-prop}.
Fig.~\ref{fig:example-intro-test-with-T-seq} shows the outcomes of testing $\rho_1$ and $\rho'_3$.
We conclude that  $\rho_1\, \|\, \rho'_3$ may pass \calT{$\mathit{seq}$}.
More general, let \calN{} and $\mathcal{N}'$ be the nets of two runs $\rho$ and $\rho'$
corresponding to two traces $\pi$ and $\pi'$, respectively.
If at least one of $\rho$ and $\rho'$ is instantiated with the concurrent run $\rho_1$, we have
$\calN{}\, \|\, \mathcal{N}' \mbox{ may pass } \calT{$\mathit{seq}$}$, 
otherwise $\calN{}\, \|\, \mathcal{N}' \mbox{ may \emph{not} pass } \calT{$\mathit{seq}$}$.
Summarizing, we have
\[
  \exists\, \rho, \rho'.\, \calN{}\, \|\, \mathcal{N}' \mbox{ may pass } \calT{$\mathit{seq}$}
\]
and even
\[
  \forall\, \rho.\, \exists\, \rho'.\, \calN{}\, \|\, \mathcal{N}' \mbox{ may pass } \calT{$\mathit{seq}$}
\]
because we can instantiate $\rho'$ with $\rho_1$, but \emph{not}
$ \forall\, \rho, \rho'\, .\, \calN{}\, \|\, \mathcal{N}' \mbox{ may pass } \calT{$\mathit{seq}$}$.
However, no \emph{must} property holds for two concurrent traces and the test \calT{$\mathit{seq}$}.
This shows that the system \calN{C} in Fig.~\ref{fig:example-systems} does not satisfy the concurrent hyperproperty $H_2$.

%
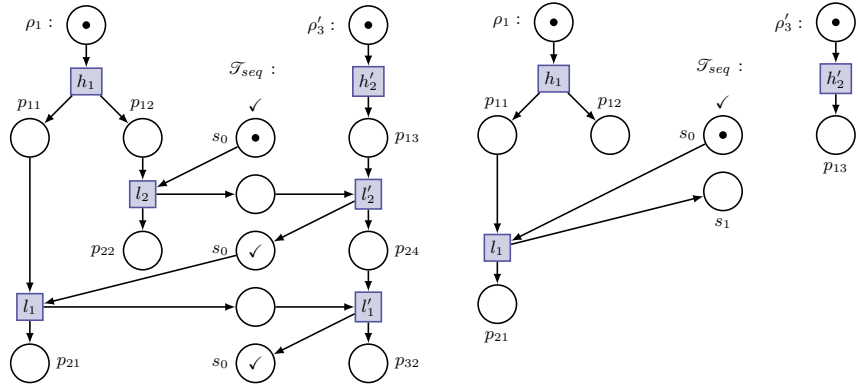
\begin{figure}[t]

\begin{minipage}{4.9cm}

\hspace{4mm}\scalebox{0.75}{
\begin{tikzpicture}

\node[place,
    tokens=1,
    label=left: $\rho_1:\ $] (p01) at (1,6) {};
    
\node[place,
    tokens=1,
    label=left: $\rho'_3:\ $] (p03) at (6,6) {};
 
\node[place,
    label=$p_{11}$] (p11) at (0,4) {};

\node[place,
    label=$p_{12}$] (p12) at (2,4) {};
    

\node[align=center] (d1) at (4,5.2) { $\calT{$\mathit{seq}$}:$ };
 
\node[place,
    tokens=1, 
    label=above:$\tick$,label=left:$s_0$] (s01) at (4,4) {};
    
\node[place,
    label=] (s21) at (4,3) {}; 
    
\node[place,
    label=left:$s_0$] (s02) at (4,2) {$\tick$};
    
\node[place,
    label=] (s11) at (4,1) {};  
    
\node[place,
    label=left:$s_0$] (s03) at (4,0) {$\tick$};
    
    
\node[place,
    label=right:$p_{13}$] (p133) at (6,4) {};
        
\node[place,
    label=right:$p_{21}$] (p21) at (0,0) {};
    
\node[place,
    label=left:$p_{22}$] (p22) at (2,2) {};
    
\node[place,
    label=right:$p_{24}$] (p24) at (6,2) {};
    
\node[place,
    label=right:$p_{32}$] (p32) at (6,0) {};

\node[transition,
    label=] (t01) at (1,5) {$h_1$};
    
\node[transition,
    label=] (t023) at (6,5) {$h'_2$};
 
\node[transition,
    label=] (t11) at (0,1) {$l_1$};

\node[transition,
    label=] (t12) at (2,3) {$l_2$};
    
\node[transition,
    label=] (t14) at (6,3) {$l'_2$};
 
\node[transition,
    label=] (t22) at (6,1) {$l'_1$};

\draw[-latex,thick] (p01) -- (t01);
\draw[-latex,thick] (p03) -- (t023);

\draw[-latex,thick] (p11) -- (t11);
\draw[-latex,thick] (p12) -- (t12);
\draw[-latex,thick] (p133) -- (t14);

\draw[-latex,thick] (p24) -- (t22);

\draw[-latex,thick] (s01) -- (t12);
\draw[-latex,thick] (s21) -- (t14);
\draw[-latex,thick] (s02) -- (t11);
\draw[-latex,thick] (s11) -- (t22);

\draw[-latex,thick] (t01) -- (p11);
\draw[-latex,thick] (t01) -- (p12);

\draw[-latex,thick] (t023) -- (p133);

\draw[-latex,thick] (t11) -- (p21);
\draw[-latex,thick] (t12) -- (p22);
\draw[-latex,thick] (t14) -- (p24);
\draw[-latex,thick] (t22) -- (p32);

\draw[-latex,thick] (t12) -- (s21);
\draw[-latex,thick] (t14) -- (s02);
\draw[-latex,thick] (t11) -- (s11);
\draw[-latex,thick] (t22) -- (s03);
 
\end{tikzpicture}
}
\end{minipage}
\begin{minipage}{1.5cm}
\hspace{1.5cm}
\end{minipage}
\begin{minipage}{2.5cm}

\vspace{-5mm}
\scalebox{0.75}{
\begin{tikzpicture}

\node[place,
    tokens=1,
    label=left: $\rho_1:\ $] (p01) at (1,6) {};
    
\node[place,
    tokens=1,
    label=left: $\rho'_3:\ $] (p03) at (6,6) {};
 
\node[place,
    label=$p_{11}$] (p11) at (0,4) {};

\node[place,
    label=$p_{12}$] (p12) at (2,4) {};
    

\node[align=center] (d1) at (4,5.2) { $\calT{$\mathit{seq}$}:$ };
 
\node[place,
    tokens=1, 
    label=above:$\tick$,label=left:$s_0$] (s01) at (4,4) {};
    
\node[place,
    label=below:$s_1$] (s21) at (4,3) {}; 
    

\node[place,
    label=below:$p_{13}$] (p133) at (6,4) {};
        
\node[place,
    label=below:$p_{21}$] (p21) at (0,1) {};

\node[transition,
    label=] (t01) at (1,5) {$h_1$};
    
\node[transition,
    label=] (t023) at (6,5) {$h'_2$};
 
\node[transition,
    label=] (t11) at (0,2) {$l_1$};

\draw[-latex,thick] (p01) -- (t01);
\draw[-latex,thick] (p03) -- (t023);

\draw[-latex,thick] (p11) -- (t11);

\draw[-latex,thick] (s01) -- (t11);

\draw[-latex,thick] (t01) -- (p11);
\draw[-latex,thick] (t01) -- (p12);

\draw[-latex,thick] (t023) -- (p133);

\draw[-latex,thick] (t11) -- (p21);

\draw[-latex,thick] (t11) -- (s21);
 
\end{tikzpicture}
}
\end{minipage}

\caption{Testing a concurrent hyperproperty with  \calT{$\mathit{seq}$}. 
  We consider the two maximal runs of the parallel composition $\rho_1\, \|\, \calT{$\mathit{seq}$} \, \|\, \rho'_3$.
  \emph{Left}: Here at first the alternative starting with $l_2$ of the test \calT{$\mathit{seq}$} is chosen.
  This runs terminates successful.
  \emph{Right}: Here at first the alternative starting with $l_1$ of \calT{$\mathit{seq}$} is chosen.
  This runs ends in a deadlock because $\rho_3$ engages first in $l_2$.
}
\label{fig:example-intro-test-with-T-seq}

\end{figure} 
%

Now we check with concurrent test \calT{$\mathit{con}$} of Fig.~\ref{fig:example-test} 
the weaker concurrent hyperproperty whether every pair of concurrent traces $\pi$ and $\pi'$ 
agrees on the occurrence of the low-security events $l_1$ and $l_2$, i.e.,
each each $l_1$ must be matched by $l'_1$, but $l_2$ may occur in between,
and vice versa for $l_2$ and $l'_2$ and a possibly intervening $l_1$.
This is property $H_1$ in Example~\ref{ex:simple-conc-hyper-prop}.
Fig.~\ref{fig:example-intro-test-with-T-con} shows the outcomes of testing $\rho_1$ and $\rho_3$.
We conclude that  $\rho_1\, \|\, \rho_3$ must pass \calT{$\mathit{con}$}.
Indeed, we have 
\[
 \forall\, \rho, \rho'\, .\, \calN{}\, \|\, \mathcal{N}' \mbox{ must pass } \calT{$\mathit{con}$}.
\]
This shows that the system \calN{C} in Fig.~\ref{fig:example-systems} satisfies the concurrent hyperproperty $H_1$.

%
\begin{figure}[t]

\begin{center}

\begin{minipage}{4.9cm}

\hspace{4mm}\scalebox{0.75}{
\begin{tikzpicture}

\node[place,
    tokens=1,
    label=left: $\rho_1:\ $] (p01) at (1,6) {};
    
\node[place,
    tokens=1,
    label=left: $\rho'_3:\ $] (p03) at (6,6) {};
 
\node[place,
    label=$p_{11}$] (p11) at (0,4) {};

\node[place,
    label=$p_{12}$] (p12) at (2,4) {};
    

\node[align=center] (d1) at (4,5.2) { $\calT{$\mathit{con}$}:$ };
 
\node[place,
    tokens=1, 
    label=above:$\tick$,label=left:$s_{02}$] (s01) at (3.5,4) {};
    
\node[place,
    label=] (s21) at (4,3) {}; 
    
\node[place,
    label=right:$s_{02}$] (s02) at (4.5,2) {$\tick$};
    
\node[place,
    tokens=1, 
    label=above:$\tick$,label=left:$s_{01}$] (s021) at (3.5,2) {};
    
\node[place,
    label=] (s11) at (4,1) {};  
    
\node[place,
    label=right:$s_{01}$] (s03) at (4.5,0) {$\tick$};
    
    
\node[place,
    label=right:$p_{13}$] (p133) at (6,4) {};
        
\node[place,
    label=right:$p_{21}$] (p21) at (0,0) {};
    
\node[place,
    label=left:$p_{22}$] (p22) at (2,2) {};
    
\node[place,
    label=right:$p_{24}$] (p24) at (6,2) {};
    
\node[place,
    label=right:$p_{32}$] (p32) at (6,0) {};

\node[transition,
    label=] (t01) at (1,5) {$h_1$};
    
\node[transition,
    label=] (t023) at (6,5) {$h'_2$};
 
\node[transition,
    label=] (t11) at (0,1) {$l_1$};

\node[transition,
    label=] (t12) at (2,3) {$l_2$};
    
\node[transition,
    label=] (t14) at (6,3) {$l'_2$};
 
\node[transition,
    label=] (t22) at (6,1) {$l'_1$};

\draw[-latex,thick] (p01) -- (t01);
\draw[-latex,thick] (p03) -- (t023);

\draw[-latex,thick] (p11) -- (t11);
\draw[-latex,thick] (p12) -- (t12);
\draw[-latex,thick] (p133) -- (t14);

\draw[-latex,thick] (p24) -- (t22);

\draw[-latex,thick] (s01) -- (t12);

\draw[-latex,thick] (s21) -- (t14);
\draw[-latex,thick] (s021) -- (t11);
\draw[-latex,thick] (s11) -- (t22);

\draw[-latex,thick] (t01) -- (p11);
\draw[-latex,thick] (t01) -- (p12);

\draw[-latex,thick] (t023) -- (p133);

\draw[-latex,thick] (t11) -- (p21);
\draw[-latex,thick] (t12) -- (p22);
\draw[-latex,thick] (t14) -- (p24);
\draw[-latex,thick] (t22) -- (p32);

\draw[-latex,thick] (t12) -- (s21);
\draw[-latex,thick] (t14) -- (s02);
\draw[-latex,thick] (t11) -- (s11);
\draw[-latex,thick] (t22) -- (s03);
 
\end{tikzpicture}
}
\end{minipage}

\end{center}

\caption{Testing a concurrent hyperproperty with  \calT{$\mathit{con}$}. 
  We consider the unique maximal run of the parallel composition $\rho_1\, \|\, \calT{$\mathit{con}$} \, \|\, \rho'_3$.
  This run terminates successfully because both concurrent components of the test end
  in a place marked with $\tick$.
}
\label{fig:example-intro-test-with-T-con}

\end{figure}
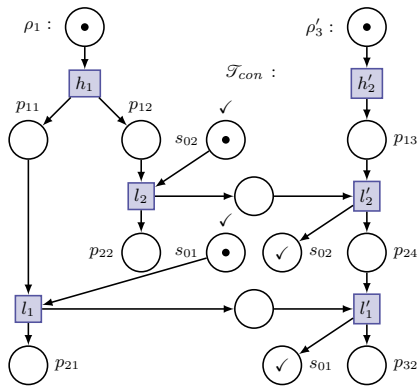 
%

\subsection{Testing the Concurrent Properties $T_1$ and $H_3$}

Consider the concurrent trace property $T_1$ of Example~\ref{ex:conc-trace-prop} for a net \calN,
where a low-security event $l$ must not depend on a high-security event $h$.
We check this by requiring that
\[
 \calN{} \mbox{ must pass } \calT{$\mathit{hl}$}
\]
for the following test \calT{$\mathit{hl}$}:

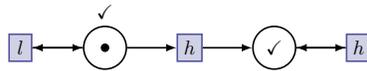
\begin{center}
\scalebox{0.75}{
  \begin{tikzpicture}

\node[transition,
    label=] (t31) at (0,3) {$l$};

\node[place,
    tokens=1,
    label=above:$\tick$] (p31) at (1.5,3) {};
    
\node[transition,
    label=] (t32) at (3,3) {$h$};

\node[place,
    label=] (p32) at (4.5,3) {$\tick$};
    
\node[transition,
    label=] (t33) at (6,3) {$h$};
    
    
%
 
 
  
 
\draw[-latex,thick] (p31) -- (t31);

\draw[-latex,thick] (p31) -- (t32);
\draw[-latex,thick] (p32) -- (t33);

 
\draw[-latex,thick] (t31) -- (p31);

\draw[-latex,thick] (t32) -- (p32);
\draw[-latex,thick] (t33) -- (p32);
 
\end{tikzpicture}
}
\end{center}

\vspace{2mm}

\noindent
This test can terminate successfully after any (possibly empty) sequence of low-security events $l$.
However, once a high-security event $h$ occurs, the test terminates successfully only after any 
(possibly empty) sequence of further $h$ events. Any low-security event $l$ occurring after the first $h$
will lead to a deadlock since the test does not offer any further synchronization on $l$.

\medskip

Finally, we consider the concurrent hyperproperty $H_3$ of generalized noninterference of Example~\ref{ex:conc-hyper-prop-gni}.
As low-security events we take $l_1, l_2 \in \Sigma_\mathit{low}$
and as high-security events $h_1, h_2 \in \Sigma_\mathit{high}$. The property is checked by requiring that
\[
 \forall\, \rho, \rho'.\ \exists \rho''.\ \calN{}\, \|\, \mathcal{N}' \|\, \mathcal{N}'' \mbox{ must pass } \calT{$\mathit{gni}$}
 \]
for the test \calT{$\mathit{gni}$} shown in Fig.~\ref{fig:test-gni}.

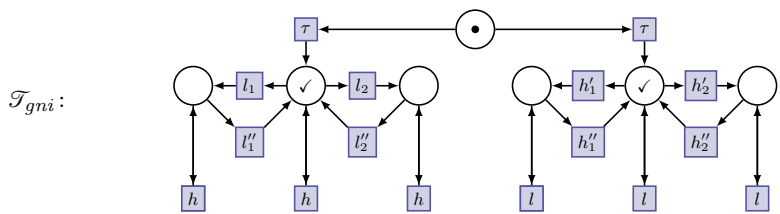
\begin{figure}[ht]

\begin{minipage}{2cm}

 \calT{$\mathit{gni}$}:   

\end{minipage}
\begin{minipage}{8cm}

\begin{center}
\scalebox{0.75}{
  \begin{tikzpicture}


\node[place,
    tokens=1,
    label=] (p54) at (5,4) {};
    
\node[transition,
    label=] (t24) at (2,4) {$\tau$};
    
\node[transition,
    label=] (t84) at (8,4) {$\tau$};


\node[place,
    label=] (p03) at (0,3) {};
    
\node[transition,
    label=] (t13) at (1,3) {$l_1$};
    
\node[place,
    label=] (p23) at (2,3) {$\tick$};
    
\node[transition,
    label=] (t33) at (3,3) {$l_2$};
    
\node[place,
    label=] (p43) at (4,3) {};

\node[transition,
    label=] (t12) at (1,2) {$l''_1$};

\node[transition,
    label=] (t32) at (3,2) {$l''_2$};
    
\node[transition,
    label=] (to01) at (0,1) {$h$};

\node[transition,
  label=] (to01) at (0,1) {$h$};

\node[transition,
    label=] (to21) at (2,1) {$h$};
    
\node[transition,
    label=] (to41) at (4,1) {$h$};


\node[place,
    label=] (p63) at (6,3) {};
    
\node[transition,
    label=] (t73) at (7,3) {$h_1'$};
    
\node[place,
    label=] (p83) at (8,3) {$\tick$};
    
\node[transition,
    label=] (t93) at (9,3) {$h_2'$};
    
\node[place,
    label=] (p103) at (10,3) {};

\node[transition,
    label=] (t72) at (7,2) {$h''_1$};

\node[transition,
    label=] (t92) at (9,2) {$h''_2$};
    
\node[transition,
    label=] (to61) at (6,1) {$l$};

\node[transition,
  label=] (to61) at (6,1) {$l$};

\node[transition,
    label=] (to81) at (8,1) {$l$};
    
\node[transition,
    label=] (to101) at (10,1) {$l$};



\draw[-latex,thick] (p54) -- (t24);
\draw[-latex,thick] (p54) -- (t84);


\draw[-latex,thick] (p23) -- (t13);
\draw[-latex,thick] (p23) -- (t33);

\draw[-latex,thick] (p03) -- (t12);

\draw[-latex,thick] (p43) -- (t32);

\draw[-latex,thick] (p03) -- (to01);
\draw[-latex,thick] (p23) -- (to21);
\draw[-latex,thick] (p43) -- (to41);


\draw[-latex,thick] (p83) -- (t73);
\draw[-latex,thick] (p83) -- (t93);

\draw[-latex,thick] (p63) -- (t72);

\draw[-latex,thick] (p103) -- (t92);

\draw[-latex,thick] (p63) -- (to61);
\draw[-latex,thick] (p83) -- (to81);
\draw[-latex,thick] (p103) -- (to101);



\draw[-latex,thick] (t24) -- (p23);
\draw[-latex,thick] (t84) -- (p83);


\draw[-latex,thick] (t13) -- (p03);
\draw[-latex,thick] (t33) -- (p43);

\draw[-latex,thick] (t12) -- (p23);
\draw[-latex,thick] (t32) -- (p23);

\draw[-latex,thick] (to01) -- (p03);
\draw[-latex,thick] (to21) -- (p23);
\draw[-latex,thick] (to41) -- (p43);


\draw[-latex,thick] (t73) -- (p63);
\draw[-latex,thick] (t93) -- (p103);

\draw[-latex,thick] (t72) -- (p83);
\draw[-latex,thick] (t92) -- (p83);

\draw[-latex,thick] (to61) -- (p63);
\draw[-latex,thick] (to81) -- (p83);
\draw[-latex,thick] (to101) -- (p103);
 
\end{tikzpicture}
}
\end{center}

\end{minipage}

\caption{Test \calT{$\mathit{gni}$}}

\label{fig:test-gni}

\end{figure}

\vspace{2mm}

\noindent
In the two universally quantified runs $\rho$ and $\rho'$,
this test uses labels $l_1, l_2, h_1, h_2$ in the net \calN{} of run $\rho$
and copies $l'_1, l'_2, h'_1, h'_2$ in the net $\mathcal{N}'$ of $\rho'$.
Likewise, in the existentially quantified run $\rho''$,
the test uses labels $l''_1, l''_2, h''_1, h''_2$ in the net $\mathcal{N}''$ of $\rho''$.

Note that the test \calT{$\mathit{gni}$} has an initial choice between the two internal $\tau$ actions,
but the conjunction in $H_3$ is modeled by must testing, which requires that for each run $\rho$ and $\rho'$
both branches terminate with a success.
In the left branch, the test is successful if it terminates when the low-security events $l_1, l_2$ are matched by corresponding events $l''_1, l''_2$, so that $F_{\mathit{low}}$ holds.
The three transitions labeled $h$ are shorthands for the occurrence of any
event $h_1, h_2, l_1', l_2', h_1', h_2',h_1'', h_2''$ that may intervene in this branch without any effect.
In the right branch, the test is successful if it terminates when the high-security events $h_1', h_2'$ are matched by corresponding events $h''_1, h''_2$, so that $G_{\mathit{high}}$ holds. The three transitions labeled $l$ are shorthands for the occurrence of any
event $l_1, l_2, h_1, h_2, l_1', l_2', l_1'', l_2''$ that may intervene in this branch without any effect.

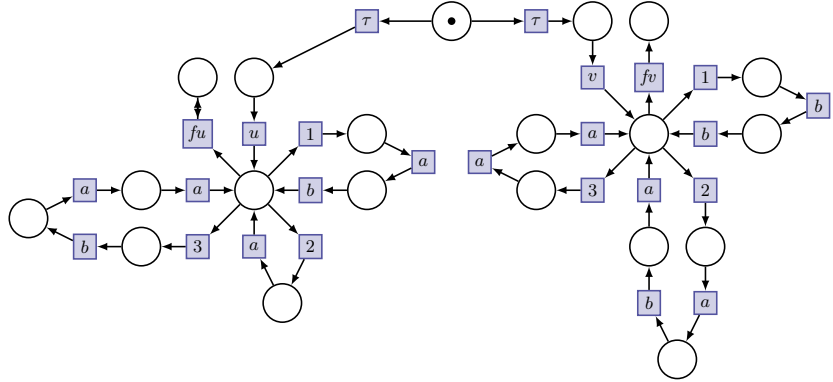
\begin{figure}[ht]
\begin{center}

\scalebox{0.75}{
  \begin{tikzpicture}


\node[place,
    tokens=1,
    label=] (p756) at (7.5,6) {};

\node[transition,
    label=] (t66) at (6,6) {$\tau$};

\node[transition,
    label=] (t96) at (9,6) {$\tau$};


 \node[place,
    label=] (p34) at (3,5) {};
    
\node[place,
     label=] (p44) at (4,5) {};
   
\node[transition,
    label=] (t33) at (3,4) {$\mathit{fu}$};
    
\node[transition,
    label=] (t43) at (4,4) {$u$};
    
\node[transition,
    label=] (t53) at (5,4) {$1$};
    
\node[place,
    label=] (p63) at (6,4) {};
    
\node[transition,
    label=] (t725) at (7,3.5) {$a$};
    
\node[transition,
    label=] (t12) at (1,3) {$a$};
    
\node[place,
    label=] (p22) at (2,3) {};

\node[transition,
    label=] (t32) at (3,3) {$a$};
    
\node[place,
    label=] (p42) at (4,3) {};
    
\node[transition,
    label=] (t52) at (5,3) {$b$};
    
\node[place,
    label=] (p62) at (6,3) {};
    
\node[place,
    label=] (p015) at (0,2.5) {};
    
\node[transition,
    label=] (t11) at (1,2) {$b$};
    
\node[place,
    label=] (p21) at (2,2) {};
    
\node[transition,
    label=] (t31) at (3,2) {$3$};
    
\node[transition,
    label=] (t41) at (4,2) {$a$};
    
\node[transition,
    label=] (t51) at (5,2) {$2$};
    
\node[place,
    label=] (p450) at (4.5,1) {};


\draw[-latex,thick] (p756) -- (t66);
    

\draw[-latex,thick] (p44) -- (t43);

\draw[-latex,thick] (p34) -- (t33);

\draw[-latex,thick] (p63) -- (t725);

\draw[-latex,thick] (p22) -- (t32);
\draw[-latex,thick] (p42) -- (t33);
\draw[-latex,thick] (p42) -- (t53);
\draw[-latex,thick] (p42) -- (t31);
\draw[-latex,thick] (p42) -- (t51);
\draw[-latex,thick] (p62) -- (t52);

\draw[-latex,thick] (p015) -- (t12);

\draw[-latex,thick] (p21) -- (t11);

\draw[-latex,thick] (p450) -- (t41);


\draw[-latex,thick] (t66) -- (p44);

\draw[-latex,thick] (t33) -- (p34);
\draw[-latex,thick] (t43) -- (p42);
\draw[-latex,thick] (t53) -- (p63);

\draw[-latex,thick] (t725) -- (p62);

\draw[-latex,thick] (t12) -- (p22);
\draw[-latex,thick] (t32) -- (p42);
\draw[-latex,thick] (t52) -- (p42);

\draw[-latex,thick] (t11) -- (p015);
\draw[-latex,thick] (t31) -- (p21);
\draw[-latex,thick] (t41) -- (p42);
\draw[-latex,thick] (t51) -- (p450);


\node[place,
     label=] (p26) at (10,6) {};
    
\node[place,
    label=] (p36) at (11,6) {};
    
\node[transition,
    label=] (t25) at (10,5) {$v$};
    
\node[transition,
    label=] (t35) at (11,5) {$\mathit{fv}$};
    
\node[transition,
    label=] (t45) at (12,5) {$1$};
    
\node[place,
    label=] (p55) at (13,5) {};
    
\node[transition,
    label=] (t645) at (14,4.5) {$b$};
    
\node[place,
    label=] (p14) at (9,4) {};
    
\node[transition,
    label=] (t24) at (10,4) {$a$};
    
\node[place,
    label=] (p34) at (11,4) {};

\node[transition,
    label=] (t44) at (12,4) {$b$};
    
\node[place,
    label=] (p54) at (13,4) {};
    
\node[transition,
    label=] (t035) at (8,3.5) {$a$};
    
\node[place,
    label=] (p13) at (9,3) {};
    
\node[transition,
    label=] (t23) at (10,3) {$3$};
    
\node[transition,
    label=] (t33) at (11,3) {$a$};
    
\node[transition,
    label=] (t43) at (12,3) {$2$};
    
\node[place,
    label=] (p32) at (11,2) {};
    
\node[place,
    label=] (p42) at (12,2) {};
    
\node[transition,
    label=] (t31) at (11,1) {$b$};
    
\node[transition,
    label=] (t41) at (12,1) {$a$};
    
\node[place,
    label=] (p350) at (11.5,0) {};


\draw[-latex,thick] (p756) -- (t96);


\draw[-latex,thick] (p26) -- (t25);

\draw[-latex,thick] (p55) -- (t645);

\draw[-latex,thick] (p14) -- (t24);
\draw[-latex,thick] (p34) -- (t35);
\draw[-latex,thick] (p34) -- (t45);
\draw[-latex,thick] (p34) -- (t23);
\draw[-latex,thick] (p34) -- (t43);
\draw[-latex,thick] (p54) -- (t44);

\draw[-latex,thick] (p13) -- (t035);

\draw[-latex,thick] (p32) -- (t33);

\draw[-latex,thick] (p42) -- (t41);

\draw[-latex,thick] (p350) -- (t31);


\draw[-latex,thick] (t96) -- (p26);

\draw[-latex,thick] (t25) -- (p34);
\draw[-latex,thick] (t35) -- (p36);
\draw[-latex,thick] (t45) -- (p55);

\draw[-latex,thick] (t645) -- (p54);

\draw[-latex,thick] (t24) -- (p34);
\draw[-latex,thick] (t44) -- (p34);

\draw[-latex,thick] (t035) -- (p14);

\draw[-latex,thick] (t23) -- (p13);

\draw[-latex,thick] (t33) -- (p34);
\draw[-latex,thick] (t43) -- (p42);

\draw[-latex,thick] (t31) -- (p32);
\draw[-latex,thick] (t41) -- (p350);

\end{tikzpicture}
}

\end{center}

\caption{Petri net \calN{$I$} simulating the input $I$ of the PCP}

\label{fig:PCP-u-or-v}

\end{figure}

\section{Decidability}
\label{sec:decidability}


Universal must testing of a net \calN{0} of the form
\[
 (*) \qquad \forall\, \rho_1, \cdots, \forall\, \rho_k.\ \calN{$1$}\, \|\, \cdots \, \|\, \calN{$k$}\ \mbox{ must pass } \calT{},
\]
can be decided because their falsification is a reachability problem.
Indeed, the negation of $(*)$ means that there exist $k$ runs of \calN{0} that composed in parallel with \calT{} 
yield a finite net in which there exist places of \calT{} without causal successor that are not marked with $\tick$.
Instead of referring to $k$ runs of \calN{0} we can equivalently refer to $k$ copies 
$\calN{$0,1$}, \dots, \calN{$0,k$}$ of \calN{0},
with suitably renamed action labels, and check the net
$\calN{} = \calN{$0,1$}\, \|\, \cdots \, \|\, \calN{$0,k$}\, \|\ \calT{}$,
with $\arrow$ as its transition relation and $\rmPl{\calT{}}$ as the set of places inside \calT{},
for the following property:
\[
 \exists\, M  \in reach(\calN{}).\ \exists\, p \in M \cap \rmPl{\calT{}}.\ p \not\in \tick
                                   \land \, \neg\, \exists\, t \in\, \arrow.\ t \mbox{ is enabled at } M.
\]
This is a reachability problem for Petri nets, which is decidable~\cite{May84}.
Since we consider safe Petri nets, this reachablity is PSPACE-complete~\cite{EspNie94}.

\medskip


By contrast, universal may testing quickly gets undecidable.

\begin{theorem}
Universal may testing is undecidable for tests with two maximal runs.
\end{theorem}

\begin{proof}
We reduce the \emph{falsification} of the Post Correspondence Problem (PCP)~\cite{Pos46}
to universal may testing using a test with two maximal runs.
\hfill $\Box$
\end{proof}

We present the proof idea for the PCP over the alphabet $\{a,b\}$.
As an input, consider the set 
\[
  I = ((u_1,v_1), (u_2,v_2), (u_3,v_3)),
\]
of pairs of subwords, where
\[
  u_1 = ab,\, v_1 = bb,\ u_2 = a,\, v_2 = aba,\ u_3 = baa,\, v_3 = aa.
\]
The PCP with this input is solvable by the correspondence $(2,3,1,3)$ because
\[
  u_2 u_3 u_1 u_3 = a\,b\,a\,a\,a\,b\,b\,a\,a = v_2 v_3 v_1 v_3.
\]
The PCP input $I$ is simulated by the Petri net \calN{$I$} shown in Fig.~\ref{fig:PCP-u-or-v}.
It consists 

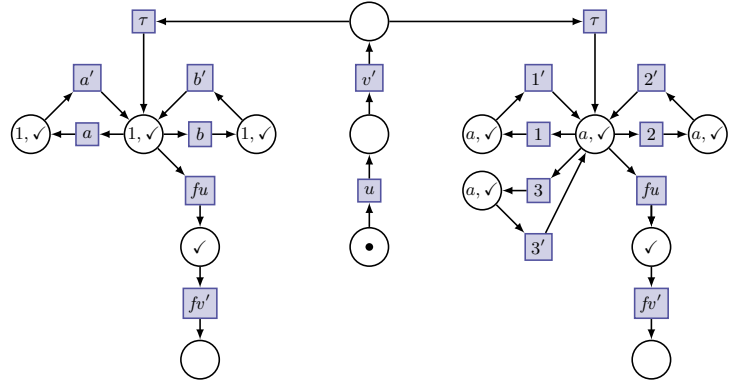
\begin{figure}[ht]

\begin{center}
\scalebox{0.75}{
  \begin{tikzpicture}

    
\node[transition,
    label=] (t43) at (4,3) {$v'$};  
    
\node[place,
    label=] (p44) at (4,4) {};  
    
\node[transition,
    label=] (t41) at (4,1) {$u$};  
    
\node[place,
    label=] (p42) at (4,2) {};  
    
\node[transition,
    label=] (t04) at (0,4) {$\tau$};  
    
\node[place,
    tokens=1,
    label=] (p40) at (4,0) {}; 
    
\node[transition,
    label=] (t84) at (8,4) {$\tau$};


\node[place,
    label=] (pm22) at (-2,2) {$1,\tick$};
    
\node[transition,
    label=] (tm13) at (-1,3) {$a'$};
    
\node[place,
    label=] (p02) at (0,2) {$1,\tick$};
    
\node[transition,
    label=] (t13) at (1,3) {$b'$};
    
\node[place,
    label=] (p22) at (2,2) {$1,\tick$};

\node[transition,
    label=] (tm12) at (-1,2) {$a$};

\node[transition,
    label=] (t12) at (1,2) {$b$};

\node[transition,
    label=] (t11) at (1,1) {$\mathit{fu}$};
    
\node[place,
    label=] (p10) at (1,0) {$\tick$};
    
\node[transition,
    label=] (t1m1) at (1,-1) {$\mathit{fv'}$};
    
\node[place,
    label=] (p1m2) at (1,-2) {};
    

\node[place,
    label=] (p62) at (6,2) {$a,\tick$};

\node[transition,
    label=] (t73) at (7,3) {$1'$};
    
\node[place,
    label=] (p82) at (8,2) {$a,\tick$};
    
\node[transition,
    label=] (t93) at (9,3) {$2'$};
    
\node[place,
    label=] (p102) at (10,2) {$a,\tick$};

\node[transition,
    label=] (t72) at (7,2) {$1$};

\node[transition,
    label=] (t92) at (9,2) {$2$};
    

\node[place,
    label=] (p61) at (6,1) {$a,\tick$};
    
\node[transition,
    label=] (t71) at (7,1) {$3$};
    
\node[transition,
    label=] (t91) at (9,1) {$\mathit{fu}$};
    
\node[place,
    label=] (p90) at (9,0) {$\tick$};

\node[transition,
    label=] (t70) at (7,0) {$3'$};
    
\node[transition,
    label=] (t9m1) at (9,-1) {$\mathit{fv'}$};
    
\node[place,
    label=] (p9m2) at (9,-2) {};



\draw[-latex,thick] (p42) -- (t43);
\draw[-latex,thick] (p40) -- (t41);

\draw[-latex,thick] (p44) -- (t04);
\draw[-latex,thick] (p44) -- (t84);


\draw[-latex,thick] (p02) -- (tm12);
\draw[-latex,thick] (p02) -- (t12);
\draw[-latex,thick] (p02) -- (t11);

\draw[-latex,thick] (pm22) -- (tm13);

\draw[-latex,thick] (p22) -- (t13);

\draw[-latex,thick] (p10) -- (t1m1);


\draw[-latex,thick] (p82) -- (t72);
\draw[-latex,thick] (p82) -- (t71);
\draw[-latex,thick] (p82) -- (t92);
\draw[-latex,thick] (p82) -- (t91);

\draw[-latex,thick] (p62) -- (t73);

\draw[-latex,thick] (p102) -- (t93);

\draw[-latex,thick] (p61) -- (t70);

\draw[-latex,thick] (p90) -- (t9m1);



\draw[-latex,thick] (t43) -- (p44);

\draw[-latex,thick] (t41) -- (p42);

\draw[-latex,thick] (t04) -- (p02);  

\draw[-latex,thick] (t84) -- (p82);


\draw[-latex,thick] (tm13) -- (p02);
\draw[-latex,thick] (t13) -- (p02);

\draw[-latex,thick] (t11) -- (p10);

\draw[-latex,thick] (tm12) -- (pm22);
\draw[-latex,thick] (t12) -- (p22);

\draw[-latex,thick] (t1m1) -- (p1m2);


\draw[-latex,thick] (t73) -- (p82);
\draw[-latex,thick] (t93) -- (p82);

\draw[-latex,thick] (t72) -- (p62);
\draw[-latex,thick] (t92) -- (p102);

\draw[-latex,thick] (t71) -- (p61); 
\draw[-latex,thick] (t91) -- (p90);

\draw[-latex,thick] (t70) -- (p82);
\draw[-latex,thick] (t91) -- (p90);

\draw[-latex,thick] (t9m1) -- (p9m2);
 
\end{tikzpicture}
}
\end{center}

\caption{Test \calT{$\mathit{PCP}$} for checking whether two runs of \calN{} do \emph{not} simulate a correspondence of the PCP.
The left branch ends in the place without $\tick$ if the runs produce letter by letter the same word, 
the right branch ends in the place without $\tick$ if the runs have chosen the same sequence of indices.}

\label{fig:test-PCP}

\end{figure}

\noindent
of two branches that are selected by an initial choice between two internal actions.
For distinguishing them in a test, the left branch starts with a transition labeled with $u$ 
and the right branch with a transition labeled with $v$.
Afterwards, their tokens reside in their center places from where they can
nondeterministically choose which of the words $u_i$ or $v_i$ for $i \in \{1,2,3\}$
to perform next. For example, the left branch simulates the subword $u_1 = ab$ by the sequence
of actions $1$, $a$, and $b$, after which the token is again on the center place
so that the next choice can be performed.
After any finite number of choices each branch may stop its activity by performing the
transition labeled with $\mathit{fu}$ or $\mathit{fv}$, respectively.

\medskip

In general, the PCP with input $I$ simulated by a net \calN{$I$} of the form above
has \emph{no} correspondence if and only if 
\[
  \forall\, \rho, \rho'\, .\ \rho\, \|\, \rho' \mbox{ may pass } \calT{$\mathit{PCP}$}
\]
for the test \calT{$\mathit{PCP}$} shown in Fig.~\ref{fig:test-PCP}.


\begin{wrapfigure}[40]{r}{0.25\textwidth}
\centering
\scalebox{0.5}{
 \begin{tikzpicture}


\node[place,
    tokens=1,
    label=] (pu1) at (1,0) {};
    
\node[transition,
    label=] (tu1) at (1,-1) {$u$};
    
\node[place,
    label=] (pu2) at (1,-2) {};
    
\node[transition,
    label=] (tu2) at (1,-3) {$2$};
    
\node[place,
    label=] (pu3) at (1,-4) {};
    
\node[transition,
    label=] (tu3) at (1,-5) {$a$};
    
\node[place,
    label=] (pu4) at (1,-6) {};
    
\node[transition,
    label=] (tu4) at (1,-7) {$3$};
    
\node[place,
    label=] (pu5) at (1,-8) {};
    
\node[transition,
    label=] (tu5) at (1,-9) {$b$};
    
\node[place,
    label=] (pu6) at (1,-10) {};
    
\node[transition,
    label=] (tu6) at (1,-11) {$a$};
    
\node[place,
    label=] (pu7) at (1,-12) {};
    
\node[transition,
    label=] (tu7) at (1,-13) {$a$};
    
\node[place,
    label=] (pu8) at (1,-14) {};
    
\node[transition,
    label=] (tu8) at (1,-15) {$1$};
    
\node[place,
    label=] (pu9) at (1,-16) {};
    
\node[transition,
    label=] (tu9) at (1,-17) {$a$};
    
\node[place,
    label=] (pu10) at (1,-18) {};
    
\node[transition,
    label=] (tu10) at (1,-19) {$b$};
    
\node[place,
    label=] (pu11) at (1,-20) {};
    
\node[transition,
    label=] (tu11) at (1,-21) {$3$};
    
\node[place,
    label=] (pu12) at (1,-22) {};

\node[transition,
    label=] (tu12) at (1,-23) {$b$};
    
\node[place,
    label=] (pu13) at (1,-24) {};
    
\node[transition,
    label=] (tu13) at (1,-25) {$a$};
    
\node[place,
    label=] (pu14) at (1,-26) {};
    
\node[transition,
    label=] (tu14) at (1,-27) {$a$};
    
\node[place,
    label=] (pu15) at (1,-28) {};

\node[transition,
    label=] (tu15) at (1,-29) {$\mathit{fu}$};
    
\node[place,
    label=] (pu16) at (1,-30) {};


\node[place,
    tokens=1,
    label=] (pv1) at (3,0) {};
    
\node[transition,
    label=] (tv1) at (3,-1) {$v'$};
    
\node[place,
    label=] (pv2) at (3,-2) {};
    
\node[transition,
    label=] (tv2) at (3,-3) {$2'$};
    
\node[place,
    label=] (pv3) at (3,-4) {};
    
\node[transition,
    label=] (tv3) at (3,-5) {$a'$};
    
\node[place,
    label=] (pv4) at (3,-6) {};
    
\node[transition,
    label=] (tv4) at (3,-7) {$b'$};
    
\node[place,
    label=] (pv5) at (3,-8) {};
    
\node[transition,
    label=] (tv5) at (3,-9) {$a'$};
    
\node[place,
    label=] (pv6) at (3,-10) {};
    
\node[transition,
    label=] (tv6) at (3,-11) {$3'$};
    
\node[place,
    label=] (pv7) at (3,-12) {};
    
\node[transition,
    label=] (tv7) at (3,-13) {$a'$};
    
\node[place,
    label=] (pv8) at (3,-14) {};
    
\node[transition,
    label=] (tv8) at (3,-15) {$a'$};
    
\node[place,
    label=] (pv9) at (3,-16) {};
    
\node[transition,
    label=] (tv9) at (3,-17) {$1'$};
    
\node[place,
    label=] (pv10) at (3,-18) {};
    
\node[transition,
    label=] (tv10) at (3,-19) {$b'$};
    
\node[place,
    label=] (pv11) at (3,-20) {};
    
\node[transition,
    label=] (tv11) at (3,-21) {$b'$};
    
\node[place,
    label=] (pv12) at (3,-22) {};

\node[transition,
    label=] (tv12) at (3,-23) {$3'$};
    
\node[place,
    label=] (pv13) at (3,-24) {};
    
\node[transition,
    label=] (tv13) at (3,-25) {$a'$};
    
\node[place,
    label=] (pv14) at (3,-26) {};
    
\node[transition,
    label=] (tv14) at (3,-27) {$a'$};
    
\node[place,
    label=] (pv15) at (3,-28) {};

\node[transition,
    label=] (tv15) at (3,-29) {$\mathit{fv'}$};
    
\node[place,
    label=] (pv16) at (3,-30) {};



\draw[-latex,thick] (pu1) -- (tu1);
\draw[-latex,thick] (pu2) -- (tu2);
\draw[-latex,thick] (pu3) -- (tu3);
\draw[-latex,thick] (pu4) -- (tu4);
\draw[-latex,thick] (pu5) -- (tu5);

\draw[-latex,thick] (pu6) -- (tu6);
\draw[-latex,thick] (pu7) -- (tu7);
\draw[-latex,thick] (pu8) -- (tu8);
\draw[-latex,thick] (pu9) -- (tu9);
\draw[-latex,thick] (pu10) -- (tu10);

\draw[-latex,thick] (pu11) -- (tu11);
\draw[-latex,thick] (pu12) -- (tu12);
\draw[-latex,thick] (pu13) -- (tu13);
\draw[-latex,thick] (pu14) -- (tu14);
\draw[-latex,thick] (pu15) -- (tu15);


\draw[-latex,thick] (pv1) -- (tv1);
\draw[-latex,thick] (pv2) -- (tv2);
\draw[-latex,thick] (pv3) -- (tv3);
\draw[-latex,thick] (pv4) -- (tv4);
\draw[-latex,thick] (pv5) -- (tv5);

\draw[-latex,thick] (pv6) -- (tv6);
\draw[-latex,thick] (pv7) -- (tv7);
\draw[-latex,thick] (pv8) -- (tv8);
\draw[-latex,thick] (pv9) -- (tv9);
\draw[-latex,thick] (pv10) -- (tv10);

\draw[-latex,thick] (pv11) -- (tv11);
\draw[-latex,thick] (pv12) -- (tv12);
\draw[-latex,thick] (pv13) -- (tv13);
\draw[-latex,thick] (pv14) -- (tv14);
\draw[-latex,thick] (pv15) -- (tv15);



\draw[-latex,thick] (tu1) -- (pu2);
\draw[-latex,thick] (tu2) -- (pu3);
\draw[-latex,thick] (tu3) -- (pu4);
\draw[-latex,thick] (tu4) -- (pu5);
\draw[-latex,thick] (tu5) -- (pu6);

\draw[-latex,thick] (tu6) -- (pu7);
\draw[-latex,thick] (tu7) -- (pu8);
\draw[-latex,thick] (tu8) -- (pu9);
\draw[-latex,thick] (tu9) -- (pu10);
\draw[-latex,thick] (tu10) -- (pu11);

\draw[-latex,thick] (tu11) -- (pu12);
\draw[-latex,thick] (tu12) -- (pu13);
\draw[-latex,thick] (tu13) -- (pu14);
\draw[-latex,thick] (tu14) -- (pu15);
\draw[-latex,thick] (tu15) -- (pu16);


\draw[-latex,thick] (tv1) -- (pv2);
\draw[-latex,thick] (tv2) -- (pv3);
\draw[-latex,thick] (tv3) -- (pv4);
\draw[-latex,thick] (tv4) -- (pv5);
\draw[-latex,thick] (tv5) -- (pv6);

\draw[-latex,thick] (tv6) -- (pv7);
\draw[-latex,thick] (tv7) -- (pv8);
\draw[-latex,thick] (tv8) -- (pv9);
\draw[-latex,thick] (tv9) -- (pv10);
\draw[-latex,thick] (tv10) -- (pv11);

\draw[-latex,thick] (tv11) -- (pv12);
\draw[-latex,thick] (tv12) -- (pv13);
\draw[-latex,thick] (tv13) -- (pv14);
\draw[-latex,thick] (tv14) -- (pv15);
\draw[-latex,thick] (tv15) -- (pv16);

\end{tikzpicture}
}
 
\caption{Maximal runs of \calN{} simulating the correspondence $(2,3,1,3)$.}

\label{fig:PCP-correspondence}


\end{wrapfigure}
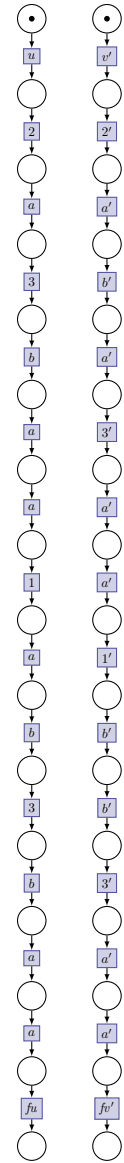


By contraposition, if the PCP has a correspondence,
there exist maximal runs $\rho$ and $\rho'$ of \calN{$I$} with nets \calN{} and $\mathcal{N}'$
such that the two maximal runs in $\calN{}\, \|\, \mathcal{N}'\, \|\, \calT{$\mathit{PCP}$}$
stemming from the two branches in \calT{$\mathit{PCP}$}
are \emph{not} sucessful, i.e., each branch ends in the unique place that is not
marked by~$\tick$.

The left branch of \calT{$\mathit{PCP}$} ends in the place without $\tick$ if $\rho$ and $\rho'$ produce letter by letter the same word.
Here the transitions labeled with unprimed symbols refer to $\rho$ and transitions labeled
with primed symbols refer to $\rho'$. The initial transitions labeled with $u$ and $v'$ ensure
that the unprimed symbols refer to the left part of \calN{$I$} simulating the $u$-part 
and that the primed symbols refer to
(the primed version of) right part of \calN{$I$} simulating the $v$-part of the proposed correspondence.
Since the correspondence is finite, this branch of the test 
ends in the place without $\tick$ after performing $\mathit{fu}$ and $\mathit{fv'}$.

The right branch of \calT{$\mathit{PCP}$} ends in the place without $\tick$ if $\rho$ and $\rho'$ have chosen the same sequence of indices
$1,2,3$ in producing the common word. Note that this branch checks the same runs $\rho$ and $\rho'$
than the left branch because $\rho$ and $\rho'$ are fixed initially.

There is one technical detail. Whereas the runs $\rho$ and $\rho'$ have no symbols in common
because $\rho$ uses only unprimed symbols and $\rho'$ only primed versions of the symbols, 
the test \calT{$\mathit{PCP}$} synchronizes in the parallel composition with $\mathcal{N}\, \|\, \mathcal{N}'$
on all its symbols except~$\tau$, i.e., on 
$a,b,a',b',u,v',\mathit{fu}, \mathit{fv'}, 1,2,3,1',2',3'$.  
To avoid unintended deadlocks we have to
enable the left branch of \calT{$\mathit{PCP}$} to be able to synchronize at every
place marked with $1$ with any transition lableled with $1,2,3,1',2'$ or $3'$,
and vice versa, the right branch of \calT{$\mathit{PCP}$} to be able to synchronize at every
place marked with $a$ with any transition lableled with $a,b,a',b',u$ or $v'$.
To enhance visibility, we dropped the loop transitions attached to these places
allowing for these synchronizations.

For the example input $I$, Fig.~\ref{fig:PCP-correspondence} shows two
maximal runs of \calN{$I$}, one with the original symbols and one with primed symbols, 
that simulate the correspondence (2,3,1,3) and 
cause the test \calT{$\mathit{PCP}$} to end for each branch in the place that is not marked 
$\tick$.

\section{Conclusion}
\label{sec:conclusion}

We introduced the notion of \emph{concurrent hyperproperties} as sets of sets
of concurrent traces. This extends classical hyperproperties, which are
sets of sets of traces. For analyzing concurrent hyperproperties,
we used Petri nets as the underlying semantic model of concurrency. 
The analysis was performed by adapting \emph{may and must testing} 
originally developed by DeNicola and Hennessy to our setting. 
Several examples illuminated the details of our approach.

As future work we envisage the introduction of suitable logics for
specifying concurrent hyperproperties, extending HyperLTL
for hyperproperties on traces (see \cite{Fin17} for an overview).
A starting point could be event structure logic~\cite{MukundT92,Pen95}.

\paragraph{Acknowledgement.}
This work was supported by the European Research Council (ERC) Grant HYPER (No. 101055412).


\bibliographystyle{./llncs/splncs04}
\bibliography{references}

\end{document}